\documentclass{amsart}%
\usepackage{amsfonts}
\usepackage{amsmath}
\usepackage{amssymb}
\usepackage{graphicx}
\usepackage[mathscr]{eucal}
\usepackage[dvipsnames,usenames]{color}
\usepackage{color}
\usepackage{colortbl}%
\usepackage{url}
\usepackage{pstricks}

\providecommand{\U}[1]{\protect\rule{.1in}{.1in}}

\newcommand{\ord}{\operatorname{ord}}

\let\scalebox\psscalebox

\newtheorem{theorem}{Theorem}
\theoremstyle{plain}

\newtheorem{definition}{Definition}
\newtheorem{example}{Example}

\newtheorem{lemma}{Lemma}

\newtheorem{proposition}{Proposition}

\numberwithin{equation}{section}
\begin{document}

\title[Linear relation on general ergodic T-Function]{Linear relation on general ergodic T-Function}
\author{Tao Shi}
\curraddr[Tao Shi]{The State Key Laboratory of Information Security\\
Institute of Software\\
Chinese Academy of Sciences \\
Beijing 100190, P.R.China\\
And Graduate University of Chinese Academy of Sciences \\
Beijing 100049, P.R.China}
\email[Tao Shi]{shitao@is.iscas.ac.cn}
\author{Vladimir Anashin}
\curraddr[Vladimir Anashin]{Institute for Information Security\\
Moscow State University, 119991, Moscow, Russia}
\email[Vladimir Anashin]{anashin@iisi.msu.ru; vladimir.anashin@u-picardie.fr}
\author{Dongdai Lin}
\curraddr[Dongdai Lin]{The State Key Laboratory of Information Security\\
Institute of Software\\
Chinese Academy of Sciences \\
Beijing 100190, P.R.China}
\email[Dongdai Lin]{ddlin@is.iscas.ac.cn}
\begin{abstract}
We find linear (as well as quadratic)  relations in a very large class of T-functions. The relations
may be used in analysis of T-function-based stream ciphers.

\end{abstract}
\maketitle


\section{{\protect\small Introduction}}

For years  linear feedback shift
registers (LFSRs) over a 2-element field $\mathbb F_2$ have been one of the most important building blocks in
keystream generators of stream ciphers. LFSRs can easily be
designed to produce binary sequences of the longest period (that is, of length $2^k-1$
for a $k$-cell LFSR over $\mathbb F_2$); LFSRs are fast and easy to
implement in hardware. However, sequences produced by LFSRs  have  linear dependencies that make easy to analyse  the sequences to construct
attacks on the whole cipher. To make output sequences of LFSRs more secure
these linear dependencies must be destroyed by a properly chosen
filter; this is the filter that carries the major cryptographical load making
the whole cipher secure.

Recently, T-functions were found to be useful tools to design fast
cryptographic primitives and ciphers based on usage on both arithmetic (addition, multiplication) and logical operations, see \cite{ASCCipher,hong05new,Hong05tsc3,KlSh,klimov03cryptographic,klimov04new,klimov05app,klimov05thesis,klsh04tfi,kotomina99,Mir1,ABCv3,tsc4Cipher,vest06phase2Cipher,vestwebsite}.
Loosely speaking, a T-function is a map of $k$-bit words into $k$-bit words such
that each $i$-th bit of image depends only on low-order bit $0,...,i$ of the
pre-image. Various methods are known to construct transitive T-functions
(the ones that produce sequences of the longest possible period, $2^k$),
see \cite{AnKhr,me:1,me:2,me:conf,me:ex,me-CJ,me-NATO,kotomina99,Lar,KlSh,klimov03cryptographic,klimov04new,klimov05thesis,hong05new}. Transitive T-functions have been considered as a candidate
to replace LFSRs in keystream generators of stream ciphers, see e.g. \cite{ABCv3,ASCCipher,Hong05tsc3,klsh04tfi,Mir1,tsc4Cipher} since sequences produced by T-function-based keystream generators
are proved to have a number
of good cryptographic properties, e.g.,  high linear and 2-adic complexity,
uniform distribution of subwords, etc., see \cite{AnKhr,Koloko,me-NATO,LinCompT-fun}. However, any word
sequence
produced by a transitive T-function has a well-known deficiency: the less
significant is the position $n$ of the bit in the word, the shorter is the period of the corresponding
bit sequence  in the output word sequence of words. To be more exact, given a transitive T-function $f$, consider a $k$-bit word sequence $
x_0, x_1,\ldots$ produced by $f$ with respect to the recurrence law
$$
x_i=f(x_{i-1})=f^i(x_0)
=\underbrace{f(\ldots(f(}_ix_0)\ldots),\qquad i=0,1,2,\ldots,
$$
(by the definition, $f^0(x_0)=x_0$);
denote $\delta_n(x_i)$ the $n$-th bit of the word $x_i$, $n=0,1,\ldots,
k-1$; then the length
of the shortest period of the bit
sequence $
\delta_n(x_0), \delta_n(x_1),\ldots$ (the \emph{$n$-th coordinate sequence})
is $2^{n+1}$. That is, only the highest
order coordinate sequence $\delta_{k-1}(x_0), \delta_{k-1}(x_1),\ldots$ reaches
the longest period, of length $2^k$. That is why the low-order coordinate
sequences are newer used to form a keystream, there either are just deleted
or serve to control other parts of the cipher. 

Moreover, the second half of the period of the coordinate sequence is just
the inverse of its first half:
\begin{equation}
\label{eq:half-per}
\delta_{n}(x_{i+2^{n}})\equiv\delta_{n}(x_{i})+1\pmod 2\text{, for
all}\ i,n=0,1,2,\ldots%
\end{equation}
Fortunately, the latter property does not cause big problems: speaking loosely,  given arbitrary
transitive T-function $f$, the half-periods $\delta_{n}(x_0),\ldots\delta_{n}(x_{2^{n}-1})$
should be considered as random and  adjacent coordinate sequences $\delta_{n-1}(x_0), \delta_{n-1}(x_1),\ldots$ and $\delta_n(x_0), \delta_n(x_1),\ldots$ as independent
(see Theorem \ref{thm:half-per} for  exact statements).

However, it was discovered that for \emph{certain} T-functions the said independence
of adjacent coordinate sequences does not take place: these sequences
satisfy linear relation of the form
\begin{equation}
\label{eq.1}
\delta_{n}(x_{i+2^{n-1}})\equiv\delta_{n}(x_{i})+\delta_{n-1}(x_{i})+z_{i}%
\pmod2,\ \text{for all}\ i=0,1,2,\ldots,%
\end{equation}
where  the length of the period of the sequence $z_i$ is only 4 (and \emph{not}
$2^n$ as in a general case, for arbitrary transitive T-function). Namely,
Molland and Helleseth in \cite{LinearKlSh05,LinPropT-func} proved this
for  a transitive T-function $f(x)=x+(x^{2}\vee C)$ suggested by Klimov and
Shamir in \cite{KlSh}; Jin-Song Wang and Wen-Feng Qi in \cite{LinPropPol} obtained
similar result for a transitive
polynomial function $f(x)=c_{0}+c_{1}x+c_{2}x^{2}+\cdots+c_{m}x^{m}$ with
integer coefficients $c_0,c_1,\ldots\in\mathbb Z=\{0,\pm 1, \pm 2, \ldots\}$.

\underline{Our contribution.} It is fourfold:
\begin{itemize}
\item First we prove that relations of type
 \eqref{eq.1}
hold for a much wider class of T-functions than polynomials over $\mathbb
Z$ and Klimov-Shamir functions $f(x)=x+(x^{2}\vee C)$, $C\in\mathbb Z$. This
wider class contains exponential T-functions (like $f(x)=3x+3^x$, fractional T-functions (like $f(x)=1+x+\frac{4}{1+2x}$) and many other T-functions that
might be extremely complex compositions of numerical and logical operators, like the following one:
\begin{equation}
\label{eq:wild}
f(x)=
\frac{x}{3}+ \left(\frac{1}{3}\right)^x+4\cdot\left(1-2\cdot\frac{\neg(x\wedge x^{2}+x^3\vee x^4)}{3 - 4\cdot(5+6x^5)^{x^6\oplus x^7}}\right)^{7+\left(\frac{8x^8}{9+10x^9}\right)}.
\end{equation}
In Theorem \ref{thm.11} below we prove that for the mentioned class of T-functions (which
is precisely defined further) relation \eqref{eq.1} holds; the length of
the period of the binary sequence $z_i$ in the relation depends on the function
$f$ and is not necessarily 4 any longer; however, it is still short.
\item Second, for a slightly narrower class of T-functions than the previous
one, we prove that a \emph{quadratic relation holds for any three consecutive coordinate
sequences}, see Theorem \ref{thm:non-lin} further. Earlier a relation of this
sort was known only for Klimov-Shamir T-function, see paper \cite{AlgStr_KlSh}
by Yong-Long Luo and Wen-Feng Qui.
\item Third, we show that both linear and quadratic relations of this
sort hold not only for univariate T-functions, but \emph{also for multi-word
T-functions and even for cascaded compositions of T-functions with other
generators}.
\item Finally we demonstrate how using  the  mentioned relations  between
coordinate sequences one can recover the rest coordinate sequences of  lower
orders 
\emph{even if  a T-function  from the mentioned class  has not been
specified}. That is, for instance, if $f$ is a
polynomial with integer coefficients, there is needless to know its coefficients
to recover  low-order coordinate sequences $(\delta_{n-2}(x_i)), (\delta_{n-3}(x_i)), \ldots$, given only a pair of coordinate
sequences $(\delta_{n}(x_i))$ and $(\delta_{n-1}(x_i))$. This is an important
conclusion since in some stream ciphers (see e.g., \cite{ABCv3,ABCv2}) coefficients of a T-function are formed during a `warming-up' stage; i.e., the coefficients
are obtained from a key and an initial vector by a special complicated procedure
and thus are not known to a cryptanalyst.
\end{itemize}
The paper serves a sort of a warning to a designer of a  T-function-based stream cipher to avoid possible flaws: both the choice of T-function and
the way it is used must guarantee that either there are no relations of this
sort among coordinate sequences or they are hidden deep enough (e.g., by
a proper filter) to prevent using them by a cryptanalyst. Even truncation
of
low-order bits may not be a remedy!

Last, but not least: we obtain our results by using  techniques of 2-adic analysis; that is, we
we expand T-functions on the whole space $\mathbb Z_2$ of 2-adic integers
and study the corresponding dynamics.
That is why we need to introduce some notions and results from 2-adic analysis
(and the 2-adic ergodic theory) before stating our results. It worth
noting here that the approach based on 2-adic dynamics (and wider, on $p$-adic
dynamics and on algebraic dynamics) recently proved its effectiveness in
various
cryptographic applications, see
corresponding monograph
\cite{AnKhr} for further details.

The paper is organized as follows:
\begin{itemize}
\item Section \ref{sec:T-non_A} concerns basics of the non-Archimedean
theory for T-functions;
\item Section \ref{sec:main-st} states our main two results (see
Appendix \ref{sec:proofs} for proofs);
\item Section \ref{sec:Ap-TfSC} discusses applications 
to T-function-based stream ciphers;
\item we conclude in Section \ref{sec:Concl}.
\end{itemize}


\section{{\protect\small The 2-adic theory of T-functions: brief survey}}
\label{sec:T-non_A}
In this section we introduce basics of what can be called a non-Archimedean
approach to T-functions. We start with a definition of a T-function and show that
T-functions can be treated as continuous
functions defined on and valued in the space of 2-adic integers. Therefore
we introduce basics of 2-adic arithmetic and of 2-adic Calculus that we will
need to state and prove our main result. There are many comprehensive monographs
on  $p$-adic numbers and $p$-adic analysis that contain all necessary definitions
and proofs, see e.g. \cite{Kobl,Mah,Sch}
or introductory chapters in \cite{AnKhr}; so further in the section we introduce
2-adic numbers in a somewhat informal manner.

It worth noting here that  the
theory of T-functions (which actually are functions that satisfy a Lipschitz
condition with a constant 1 w.r.t. 2-adic metric) was developed   by mathematicians
during decades prior to first publication of Klimov and Shamir on T-functions
\cite{KlSh} in 2003,
and in a much more general setting, for arbitrary prime $p$, and
not only for $p=2$. Moreover, various criteria of invertibility and single
cycle property of T-functions were obtained within $p$-adic ergodic theory
(see e.g. \cite{me:conf,me:1})
nearly a decade prior to the first publication
of Klimov and Shamir on T-functions \cite{KlSh}: Actually
a T-function $f$ is invertible if and only if it preserves Haar measure on 2-adic
integers, and $f$ has a single cycle property  if and only if it is ergodic
w.r.t. the Haar measure.  Unfortunately, cryptographic community were not aware of  that
work done by mathematicians although in
various papers there was directly pointed out that these functions  might be
useful to cryptography, see  e.g. \cite{me:1,me:conf,me:ex,me:2}.
To the
moment, there exists a well developed mathematical discipline, the $p$-adic
ergodic theory, a part of the non-Archimedean dynamics, and various crucial
cryptographic  properties
of T-functions can be studied, properly understood and explained within this
theory. Moreover, the theory has a well-developed tools to study cascaded
compositions that include T-functions along with other standard cryptographic
primitives (e.g., LFSRs): the compositions can be treated as wreath products
of dynamical systems, and single cycle property of the composition is just
ergodicity of the corresponding dynamical system, the wreath product.
So the present paper serves an example of  how effective are tools of the mentioned theory in a study of concrete cryptographical properties.
For further reading on the theory as well as on its applications to
cryptography (and to other sciences) readers are referred to monograph \cite{AnKhr}.

\subsection{T-functions}
An $n$-variate T-function is a mapping
\begin{equation}
\label{eq:t-mult}
\left(\alpha_0^{\downarrow},\alpha_1^{\downarrow},
\alpha_2^{\downarrow},\ldots\right)\mapsto
\left(\Phi_0^{\downarrow}\left(\alpha_0^{\downarrow}\right),
\Phi_1^{\downarrow}\left(\alpha_0^{\downarrow},
\alpha_1^{\downarrow}\right),
\Phi_2^{\downarrow}\left(\alpha_0^{\downarrow},
\alpha_1^{\downarrow},\alpha_2^{\downarrow}\right),
\ldots\right),
\end{equation}
where $\alpha_i^{\downarrow}\in\mathbb F_2^n$ is a Boolean columnar
$n$-dimensional vector over a 2-element field $\mathbb F_2=\{0,1\}$, and
$$
\Phi_i^{\downarrow}\colon (\mathbb F_2^n)^{i+1}\to \mathbb F_2^m
$$
maps $(i+1)$ Boolean columnar $n$-dimensional vectors
$\alpha_0^{\downarrow},\ldots,\alpha_i^{\downarrow}$ to
$m$-dimensional columnar Boolean vector
$\Phi_i^{\downarrow}\left(\alpha_0^{\downarrow},
\ldots,\alpha_i^{\downarrow}\right)$. Accordingly, a
{univariate T-function $f$} is a mapping
\begin{equation}
\label{eq:t-uni}
(\chi_0;\chi_1;\chi_2; \ldots)\stackrel{f}{\mapsto}
(\psi_0(\chi_0);\psi_1(\chi_0,\chi_1);
\psi_2(\chi_0,\chi_1,\chi_2);\ldots),
\end{equation}
where $\chi_j\in\{0,1\}$, and each
$\psi_j(\chi_0,\ldots,\chi_j)$ is a Boolean function in
Boolean variables $\chi_0,\ldots,\chi_j$. $T$-functions may
be viewed as mappings from non-negative integers to
non-negative integers: e.g., a univariate $T$-function $f$
sends a number with the base-$2$ expansion
$$
\chi_0+\chi_1\cdot 2+\chi_2\cdot 2^2+\cdots
$$
to the number with the base-2 expansion
$$
\psi_0(\chi_0)+\psi_1(\chi_0,\chi_1)\cdot
2+\psi_2(\chi_0,\chi_1,\chi_2)\cdot 2^2+\cdots
$$
Further in the paper we refer to these Boolean functions
$\psi_0, \psi_1,\psi_2,\ldots$ as \emph{coordinate
functions} of a $T$-function $f$.
If we restrict
$T$-functions to the set of all numbers whose base-$2$
expansions are not longer than $k$, we sometimes refer to these
restrictions as \emph{$T$-functions on $k$-bit words}: We usually associate
the set of all $k$-bit words to the set  $\{0,1,\ldots, 2^k-1\}$ of all residues modulo $2^k$; the latter set constitutes the residue ring $\mathbb Z/2^k\mathbb
Z$ modulo $2^k$ w.r.t. modulo $2^k$ operations of addition and multiplication.

The determinative property of T-functions (which might be used to state
equivalent
definition of a T-function) is \emph{compatibility} with all congruences
modulo powers of 2: Given a (univariate) T-function $f$,
\begin{equation}
\label{eq:compat}
\text{if}\ a\equiv b\pmod {2^s}\ \text{then}\ f(a)\equiv f(b) \pmod{2^s}.
\end{equation}
Vice versa, every compatible map is a T-function.

Important examples of $T$-functions are basic machine
instructions:
\begin{itemize}
\item integer arithmetic operations (addition,
multiplication,\ldots);
\item bitwise logical operations ($\vee$, $\oplus$, $\wedge$,
$\neg$);
\item some their compositions (masking, shifts towards high
order bits, reduction modulo $2^k$).
\end{itemize}
Since obviously a composition of T-functions is a
T-function (for instance, \emph{any polynomial with
integer coefficients is a T-function}), the T-functions
are natural functions that can be evaluated by  digital computers. 

\subsection{2-adic numbers and 2-adic Calculus}
\label{anashin:sec:NA}

As it follows directly from the definition, any T-function
is well-defined on the set $\mathbb Z_2$ of all infinite binary
sequences $\ldots\delta_2(x)\delta_1(x)\delta_0(x)=x$, where
$\delta_j(x)\in\{0,1\}$, $j=0,1,2,\ldots$. Arithmetic
operations (addition and multiplication) with these
sequences could be defined via standard ``school-textbook''
algorithms of addition and multiplication of natural numbers
represented by base-$2$ expansions. Each term of a sequence
that corresponds to the sum (respectively, to the product)
of two given sequences could be calculated by these
algorithms within a finite number of steps.

Thus, $\mathbb Z_2$ is a commutative ring with respect to the so
defined addition and multiplication. The ring $\mathbb Z_2$ is
called the ring of \emph{$2$-adic integers}. The ring $\mathbb Z_2$
contains a subring $\mathbb Z$ of all rational integers: For
instance, $\ldots111=-1$, since
$$
\arraycolsep0pt\renewcommand{\arraystretch}{0}
\begin{array}{@{\protect\vphantom{a_0^0}}rccccc}
&\ldots&1&1&1&1\\
\multicolumn1r{\raisebox{-.5\height}[0pt][0pt]{$+$ }}&&&&&\\
&\ldots&0&0&0&1\\\cline{2-6}&\ldots&0&0&0&0
\end{array}
$$

Moreover, the ring $\mathbb Z_2$ contains all rational numbers that
can be represented by irreducible fractions with odd
denominators. For instance, the following calculations show
that $\ldots01010101\times \ldots00011=\ldots111$, i.e.,
that $\ldots01010101=-1/3$ since $\ldots00011=3$ and
$\ldots111=-1$:

$$
\arraycolsep0pt\renewcommand{\arraystretch}{0}
\begin{array}{@{\protect\vphantom{a_0^0}}rccccccc}
&\ldots&0&1&0&1&0&1\\
\multicolumn1r{\raisebox{-.5\height}[0pt][0pt]{$\times$
}}&&&&&&&\\ &\ldots&0&0&0&0&1&1\\\cline{2-8}
&\ldots&0&1&0&1&0&1\\\multicolumn1r{\raisebox{-.5\height}[0pt][0pt]{$+$
}}&&&&&&&\\&\ldots&1&0&1&0&1&\\\cline{2-8}
&\ldots&1&1&1&1&1&1
\end{array}
$$

Sequences with only finite number of $1$s correspond to
non-negative rational integers in their base-$2$ expansions,
sequences with only finite number of $0$s correspond to
negative rational integers, while eventually periodic
sequences (that is, sequences that become periodic starting
with a certain place) correspond to rational numbers
represented by irreducible fractions with odd denominators:
For instance, $3=\ldots00011$, $-3=\ldots11101$,
$1/3=\ldots10101011$, $-1/3=\ldots1010101$. So the $j$-th
term $\delta_j(u)$ of the corresponding sequence $u\in\mathbb Z_2$
is merely the $j$-th digit of the base-$2$ expansion of $u$
whenever $u$ is a non-negative rational integer,
$u\in\mathbb N_0=\{0,1,2,\ldots\}$.

What is important, the ring $\mathbb Z_2$ is a metric space with
respect to the metric (distance) $d_2(u,v)$ defined by the
following rule: $d_2(u,v)=\|u-v\|_2=1/2^n$, where $n$ is the
smallest non-negative rational integer such that
$\delta_n(u)\ne\delta_n(v)$, and $d_2(u,v)=0$ if no such $n$
exists (i.e., if $u=v$). For instance $d_2(3,1/3)=1/8$. The
function $d_2(u,0)=\|u\|_2$ is the \emph{2-adic absolute value} of the $2$-adic integer
$u$, and $\ord_2 u=-\log_2\|u_2\|_2$ is the $2$-adic valuation
of $u$. Note that for $u\in\mathbb N_0$ the valuation $\ord_2 u$ is
merely the exponent of the highest power of $2$ that divides
$u$ (thus, loosely speaking, $\ord_2 0=\infty$, so
$\|0\|_2=0$).

Now we can represent every 2-adic integer $x=\ldots\delta_2(x)\delta_1(x)\delta_0(x)$
(where
$\delta_i(x)\in\{0,1\}$, $i=0,1,2,\ldots$) as the series
\begin{equation}
\label{eq:can}
x=\sum_{i=0}^\infty\delta_i(x)\cdot 2^i; \ \ (\text{where}\ \delta_i(x)\in\{0,1\}, i=0,1,2,\ldots).
\end{equation}
The series in the right-hand side are called \emph{canonical 2-adic expansion}
of the 2-adic integer $x$; the series converges to $x$ with respect to the
2-adic metric.

Although T-functions are maps from 2-adic integers to 2-adic integers,
we also introduce here 2-adic numbers whic are not necessarily 2-adic integers.
Denote $\mathbb Q_2$ the set of all series of the form $u=\sum_{i=-k}^\infty\alpha_i\cdot
2^i$ for all $k=0,1,2,\ldots$ and all $\alpha_{-k},\alpha_{-k+1},\ldots\in\{0,1,\}$.
In a way similar to that we have defined addition and multiplication on $\mathbb Z_2$,
we define these operations on $\mathbb Q_2$; the set $\mathbb Q_2$ with respect
to the so defined addition and multiplication is a \emph{field of 2-adic
numbers}, whereas $\mathbb Z_2$ is a ring of integers of this field. The absolute value $\|\cdot\|_2$ can be expanded to the whole field $\mathbb
Q_2$ (by setting $\|u\|_2=2^{-\ell}$ where $\ell$ is the smallest of $j=-k,-k+1,\ldots$
such that $\alpha_j\ne 0$); so $\mathbb Q_2$ is a metric space, and the 2-adic
absolute
value $\|\cdot\|_2$
satisfy all usual axioms. In particular, given $a,b,c\in\mathbb Q_2$,
\begin{enumerate}
\item $\|a\cdot b\|_2=\|a\|_2\cdot\|b\|_2$,
\item $\|a- c\|_2\le\|a-b\|_2+\|b-c\|_2$ (the \emph{triangle inequality}).
\end{enumerate}
It worth noting here that for the 2-adic metric the triangle inequality actually
holds in a stronger form:
$$
\|a- c\|_2\le\max\{\|a-b\|_2,\|b-c\|_2\}\  (\text{the \emph{strong} triangle inequality}),
$$
for all $a,b,c,\in\mathbb Q_2$. Now metric on the $n$-th Cartesian power $\mathbb Q_2^n$ of $\mathbb
Q_2$ can be defined in
the following way:
$
\|(a_1,\ldots,a_n)-(b_1,\ldots,b_n)\|_2=\max\{\|a_i-b_i\|_2\colon i=1,2,\ldots,n\}
$
for every $(a_1,\ldots,a_n),(b_1,\ldots,b_n)\in\mathbb Q_2^n$.

Once the metric is defined, one defines notions of
convergent sequences, limits, continuous functions on the
metric space, and derivatives if the space is a commutative
ring. For instance, with respect to the 2-adic metric the following sequence tends to $-1$:
$$
1,3,7,15,31,\ldots,2^n-1,\ldots\xrightarrow[d_2]{}-1.
$$
Derivations of a function $f\colon\mathbb Z_2\to\mathbb Z_2$, which is defined
on and valuated in the space $\mathbb Z_2$ of 2-adic integers, may be defined
in a standard way as in classical (e.g., real) Calculus just by replacing
real absolute value $|\cdot|$ by the 2-adic absolute value $\|\cdot\|_2$,
as follows:
\begin{definition}[2-adic differentiability]
\label{def:2-diff}
The function $f$ is said to be differentiable at the point $x\in\mathbb Z_2$
\textup{(}and the 2-adic number $f^\prime(x)\in\mathbb Q_2$ is said to be its derivative
at the point $x$\textup{)}
if and only if
for  arbitrary $M\in\mathbb N=\{1,2,\ldots\}$ and sufficiently small (w.r.t. the 2-adic absolute value) $h$ the following inequality holds:
$$\left\|\frac{f(x+h)-f(x)}{h}-f^\prime
(x)\right\|_2\le\frac{1}{2^M}$$
\end{definition}

Reduction modulo $2^n$ of a $2$-adic integer $v$, i.e.,
setting all terms of the corresponding sequence with indexes
greater than $n-1$ to zero (that is, taking the first $n$
digits in the representation of $v$) is just an
approximation of a $2$-adic integer $v$ by a rational
integer with precision $1/2^n$: This approximation is an
$n$-digit positive rational integer $v \wedge (2^n-1)$; the
latter will be denoted also as $v\bmod{2^n}$.

Actually \emph{a processor works with approximations of
$2$-adic integers with respect to $2$-adic metric}: When an
overflow happens, i.e., when a number that must be written
into an $n$-bit register consists of more than $n$
significant bits, the processor just writes only $n$ low
order bits of the number into the register thus reducing the
number modulo $2^n$. Thus, precision of the approximation is
defined by the bitlength of the processor.
\subsection{2-adic continuity of T-functions}
\label{ssec:T-Calc}
What is most important within the scope of the paper is that
all T-functions are \emph{continuous} functions of
2-adic variables since \emph{all T-functions satisfy
{Lipschitz condition with a constant $1$} with respect to
the 2-adic metric}, and vice versa.

Indeed, it is obvious that the function $f\colon
\mathbb Z_2\to\mathbb Z_2$ satisfy the condition
$\|f(u)-f(v)\|_2\le\|u-v\|_2$ for all $u,v\in\mathbb Z_2$ if and
only if $f$ is compatible, since the inequality
$\|a-b\|_2\le1/2^k$ is just equivalent to the congruence
$a\equiv b \pmod{2^k}$. A similar property holds for
$n$-variate T-functions (we just use the metric $\|\cdot\|_2$ on the $n$-Cartesian
power $\mathbb Z_2^n$).
So we conclude:
\begin{center}\em
T-functions${}={}$compatible functions${}=1$-Lipschitz
functions
\end{center}
This implies in particular that given a T-function $f\colon\mathbb Z_2\to\mathbb
Z_2$ and $n\in\mathbb N$, the map $f\bmod 2^n\colon z\mapsto f(z)\bmod 2^n$
is a well-defined transformation of the residue ring $\mathbb Z/2^n\mathbb
Z=\{0,1,\ldots,2^n-1\}$; actually the reduced map $f\bmod 2^n$ is a T-function
on $n$-bit words.

The observation we just have made indicates why the
 the $2$-adic analysis can be used
in a study of T-functions. For instance, one can prove
that the following functions satisfy Lipschitz condition
with a constant $1$ and thus are T-functions (and so also
be used  in compositions of  cryptographic primitives):
\begin{itemize}
\item subtraction: $(u,v)\mapsto u-v$;
\item exponentiation: $(u,v)\mapsto (1+2u)^v$;
\item raising to negative powers, $u\mapsto(1+2u)^{-n}$;
\item division: $(u,v)\mapsto \frac{u}{1+2v}$.
\end{itemize}

We now consider derivations of T-functions. We first note that as a T-function
is mere a 1-Lipschitz function w.r.t. 2-adic metric, once the derivative
exists, \emph{the derivative must be a 2-adic integer}. That is, for
the case of  T-functions we can re-state
Definition \ref{def:2-diff} in the following equivalent form:
\begin{definition}[differentiability of T-functions]
\label{def:2-diff_T}
A T-function $f\colon\mathbb Z_2\to\mathbb Z_2$ is said to be differentiable at the point $x\in\mathbb Z_2$
\textup{(}and the 2-adic number $f^\prime(x)\in\mathbb Z_2$ is said to be its derivative
at the point $x$\textup{)}
if and only if
for  arbitrary $M\in\mathbb N=\{1,2,\ldots\}$ and sufficiently small \textup{(w.r.t. the 2-adic absolute value)} $h\in\mathbb Z_2$ the following congruence holds:
$$f(x+h)\equiv f(x)+f^\prime(x)\cdot h\pmod{2^{\ord_2h+M}}$$

\end{definition}
\begin{example}[differentiability of $\wedge$]
The function $f(x)=x\wedge c$ is differentiable at every  $x\in\mathbb Z_2$ for any $c\in \mathbb Z$, and
$$
f^\prime(x)=
\begin{cases}
0,& \text{if $c\ge 0$};\\
1,& \text{if $c<0$}.
\end{cases}
$$
\end{example}
\begin{proof}
Indeed, take $n$ greater than the bitlength of $|c|$ (that is, $n\ge\log_2|c|+1$); then for all $s\in\mathbb Z_2$:

$$
f(x+2^ns)=
\begin{cases}
f(x)&, \text{if $c\ge 0$},\\
f(x)+2^ns&, \text{if $c<
0$},
\end{cases}
$$

\end{proof}
In the same manner we can fill the rest of the table of derivations of logical
T-functions:
\begin{example}[derivations of other logical T-functions]
Let $c\in\mathbb Z$, then for every $x\in\mathbb Z_2$
\begin{equation}
(\neg x)^\prime=-1;\qquad
(x\oplus c)^\prime=
\begin{cases}
~~1,& \text{if $c\ge 0$};\\
-1,& \text{if $c<0$}.
\end{cases}
\qquad
(x\vee c)^\prime=
\begin{cases}
1,& \text{if $c\ge 0$};\\
0,& \text{if $c<0$}.
\end{cases}
\end{equation}
\end{example}
Note that rules of derivations (e.g., chain rule) do not depend on metric;
thus they are the same both in a classical and in a 2-adic cases, so
applying the rules one can find derivatives of T-functions that are used in stream
ciphers:
\begin{example}[derivative of the Klimov-Shamir T-function]
$$
(x+(x^2\vee 5))^\prime=1+2x
$$
\end{example}
Now with the use of Definition \ref{def:2-diff_T} we define the notion of
\emph{uniform differentiability} of a T-function in the same way as in
classical Calculus:
\begin{definition}[uniform differentiability]
A T-function $f\colon\mathbb Z_2\to\mathbb Z_2$ is called uniformly differentiable \textup{(or, equidifferentiable)} 
iff for every sufficiently large  $M\in\mathbb N$ there exists $K\in\mathbb N$
such that once $|h|_2\leqslant\frac{1}{2^K}$ \textup{(that is, once $h\equiv
0\pmod{2^K}$)},  the congruence
$$f(x+h)\equiv f(x)+f^\prime(x)\cdot h\pmod{2^{\ord_2h+M}}$$
holds for all $x\in\mathbb Z_2$. Given $M$, the minimum  $K=K(M)$ with this property
is denoted via $N_M(f)$.
\end{definition}
For instance, it can be easily verified that Klimov-Shamir T-function $f(x)=x+(x^2\vee
5)$ is uniformly differentiable and $N_M(f)=M$.

Now we introduce another notion related to differentiability that has no
direct analogs in classical Calculus.
\begin{definition}[differentiability modulo $2^M$]
\label{def:diff_mod}
Given $M\in\mathbb N$, a T-function $f\colon\mathbb Z_2\to\mathbb Z_2$ is said to be differentiable
modulo $2^M$ at the point $x\in\mathbb Z_2$
\textup{(}and the 2-adic integer $f_M^\prime(x)\in\mathbb Z_2$ is said to be its derivative modulo $2^M$
at the point $x$\textup{)}
if and only if
for a  
sufficiently small \textup{(w.r.t. the 2-adic absolute value)} $h\in\mathbb Z_2$ the following congruence holds:
$$f(x+h)\equiv f(x)+f^\prime(x)\cdot h\pmod{2^{\ord_2h+M}}.$$

\end{definition}
\begin{definition}[uniform differentiability modulo $2^M$]
\label{def:unidiff_mod}
Given $M\in\mathbb Z$, a T-function $f\colon\mathbb Z_2\to\mathbb Z_2$ is called uniformly differentiable modulo $2^M$ 
iff 
there exists $K\in\mathbb N$
such that once $|h|_2\leqslant\frac{1}{2^K}$ \textup{(that is, once $h\equiv
0\pmod{2^K}$)},  the congruence
$$f(x+h)\equiv f(x)+f^\prime(x)\cdot h\pmod{2^{\ord_2h+M}}$$
holds for all $x\in\mathbb Z_2$. The minimum  $K=K(M)$ with this property
is denoted via $N_M(f)$.
\end{definition}
Note that the notion of derivative modulo $2^M$ is somewhat like saying
`a derivative with a precision of $M$ digits after the point' in classical
Calculus; however, the latter in  real Calculus is meaningless, whereas
in 2-adic Calculus the phrase has a precise mathematical meaning.

From  Definition \ref{def:diff_mod} it readily follows that the  derivative modulo $2^M$ is defined up to a summand
which is 0 modulo $2^M$; that is, if a T-function $f\colon\mathbb Z_2\to\mathbb
Z_2$ is uniformly differentiable
modulo $2^M$ then its derivative modulo $2^M$ is a map from $\mathbb Z_2$
into the residue ring $\mathbb Z/2^M\mathbb Z$. Furthermost, it can be proved
(see \cite{AnKhr}) that a derivative modulo $2^m$ is a periodic function with
a period of length $2^{N_M(f)}$. Thus we state
\begin{proposition}[derivatives modulo $2^M$]
\label{prop:diff_mod}
If a T-function $f$ is uniformly differentiable modulo $2^M$, then its derivative
modulo $2^M$   is a periodic function with
a period of length $2^{N_M(f)}$; so the derivative can be considered as a map from the residue ring $\mathbb Z/2^{N_M(f)}\mathbb
Z$ to the  residue ring $\mathbb Z/2^M\mathbb Z$.
\end{proposition}
Rules of derivation modulo $2^M$ are of a similar form to that of the classical
case; however, they are congruences modulo $2^M$ rather than equalities.
\begin{example}
The T-function $f(x)=x\oplus(-1/3)$ is uniformly differentiable modulo $2^M$
if and only if $M=1$; its derivative modulo $2$ is 1, and $N_2(f)=1$. If $M>1$ then $f$ is differentiable modulo $2^M$ at
no point. 
\end{example}
From Definition \ref{def:diff_mod} it immediately  follows that
\begin{itemize}
\item if a T-function
is differentiable modulo $2^{M+1}$ then it is uniformly differentiable
modulo $2^M$;
\item a  T-function is uniformly differentiable iff it is uniformly differentiable modulo
$2^M$ for all $M\in\mathbb N$.
\end{itemize}
Thus, we have the following hierarchy of classes of uniform differentiability:
$$
\mathfrak D_1\supset\mathfrak D_2\supset\mathfrak D_3\supset\cdots\supset\mathfrak
D_\infty,
$$
where
$\mathfrak D_i$ is the class of all T-functions that are uniformly differentiable
modulo $2^i$, $i=1,2,3,\ldots$, and $\mathfrak D_\infty$ is a class of all uniformly differentiable T-functions.
It turns out that \emph{the T-functions of most interest to cryptography, the ones
that are invertible, all lie in $\mathfrak D_1$; that is, they all are uniformly differentiable modulo 2}.
\subsection{Differentiability, invertibility and single cycle property}
Given $n\in\mathbb N$, a T-function $f\colon\mathbb Z_2\to\mathbb Z_2$ is said to be \emph{bijective
modulo $2^n$} iff it is invertible on $n$-bit words; that is, iff the reduced
map $f\bmod 2^n\colon\mathbb Z/2^n\mathbb Z\to\mathbb Z/2^n\mathbb Z$ is
a permutation on the residue ring $\mathbb Z/2^n\mathbb Z$. Similarly, a
T-function $f\colon\mathbb Z_2\to\mathbb Z_2$ is said to be \emph{transitive
modulo $2^n$} iff it is a single cycle on $n$-bit words; that is, iff the reduced
map $f\bmod 2^n\colon\mathbb Z/2^n\mathbb Z\to\mathbb Z/2^n\mathbb Z$ is
a permutation on the residue ring $\mathbb Z/2^n\mathbb Z$ with the only
cycle (hence, with the cycle of length $2^n$).
\begin{definition}
\label{def:trans}
We say that a T-function $f\colon\mathbb Z_2\to\mathbb Z_2$ is bijective
iff it is bijective modulo $2^n$ for all $n\in\mathbb N$; we say that $f$
is transitive iff $f$ is transitive modulo $2^n$ for all $n\in\mathbb N$.
\end{definition}
Actually the above definition is a theorem that is proved in the $p$-adic
ergodic theory: transitive T-functions are exactly 1-Lipschitz ergodic transformations
on $\mathbb Z_2$, whereas bijective T-functions are measure-preserving isometries
of $\mathbb Z_2$ (see \cite{AnKhr}). For not to overload the paper we are not going to give a deeper look into the $p$-adic ergodic theory; within the
scope of the paper the above definition is sufficient. The point is that
for some T-functions bijectivity (resp., transitivity) modulo $2^n$ for \emph{some}
$n\in\mathbb N$ implies their bijectivity (resp., transitivity); that is,
under certain conditions, if a T-function is  invertible (resp., has a single
cycle property) on $n$-bit words for \emph{some}
$n\in\mathbb N$, then it is bijective (resp, transitive) invertible (resp., has a single
cycle property) on $n$-bit words  for \emph{all} $n\in\mathbb N$. For
proofs of rest claims of the section readers are referred to monograph \cite{AnKhr}.
\begin{proposition}
\label{prop:bij_u-diff}
If a T-function $f\colon\mathbb Z_2\to\mathbb Z_2$ is bijective then it is
uniformly differentiable modulo 2 and its derivative modulo 2 is 1 everywhere:
$f^\prime_2(x)\equiv 1\pmod 2$ for all $x\in\mathbb Z_2$ \textup{(equivalently, for
all $x\in\mathbb Z/2^{N_1(f)}\mathbb Z$)}.
\end{proposition}
\begin{theorem}
Let a T-function $f$ be uniformly differentiable modulo 2.  Then
$f$ is bijective iff $f$ is bijective modulo $2^{N_1(f)}$ and $f^\prime_2(x)\equiv 1\pmod 2$ everywhere. Equivalently: if and only if $f$ is bijective
modulo $2^{N_1(f)+1}$.
\end{theorem}
\begin{theorem}
\label{thm:erg}
Let a T-function $f$ be uniformly differentiable modulo 4.  Then
$f$ is transitive iff $f$ is transitive modulo $2^{N_2(f)+2}$.
\end{theorem}
\begin{example}
The Klimov-Shamir T-function $f(x)=x+(x^2\vee 5)$ is transitive.
\end{example}
\begin{proof}
Indeed,  $f$ is uniformly differentiable, $N_2(f)=2$; so it suffices to check
whether the residues modulo 16 of  $0, f(0), f^2(0)=f(f(0)),\ldots, f^{15}(0)$
are all different. This can readily be verified by direct calculations.
\end{proof}
It worth noting here that all transitive (as well as all bijective) T-functions can
be represented in a certain `explicit' form:
\begin{theorem}[\protect{\cite{me:2}, also \cite[Theorem 4.44]{AnKhr}}]
\label{thm:expl}
~
\begin{itemize}
\item A T-function $f\colon\mathbb Z_2\to\mathbb Z_2$ is bijective if
and only if it is of the form $f(x)=c+x+2g(x)$, where $g$ is an arbitrary T-function, $c\in\{0,1\}$.
\item A T-function $f\colon\mathbb Z_2\to\mathbb Z_2$ is transitive if
and only if it is of the form $f(x)=1+x+2(g(x+1)-g(x))$, where $g$ is an arbitrary T-function.
\end{itemize}
\end{theorem}
\subsection{Properties of coordinate sequences} Given a transitive T-function
$f\colon\mathbb Z_2\to\mathbb Z_2$ and a 2-adic integer $x_0\in\mathbb Z_2$,
consider $i$-th  coordinate sequence $(\delta_i(f^j(x_0))_{j=0}^\infty$. The sequence satisfies recurrence
relation \eqref{eq:half-per}; that is, the  second half of the period
of the $i$-th coordinate sequence is a bitwise negation of the first half;
so the  shortest period  (which is of length  $2^{i+1}$) of the sequence
is completely determined by its  first $2^i$ bits. It turns out that given
\emph{arbitrary} T-function $f$, the first half's of periods of coordinate
sequences should be considered as independent, in the following meaning:
\begin{theorem}[The independence of coordinate sequences]
\label{thm:half-per}
 Given a set $\mathcal S_0,\mathcal S_1, \mathcal S_2, \ldots$ of binary sequences $\mathcal S_i=(\zeta_j)_{j=0}^{2^i-1}$ of length $2^i$, $i=0,1,2,\ldots$, there exists a transitive T-function
$f$ and a 2-adic integer $x_0\in\mathbb Z_2$
such that each first half of each $i$-th coordinate sequence is the sequence $\mathcal
S_i$, $i=0,1,2,\ldots$:
$$
\delta_i(f^j(x_0))=\zeta_j, \quad \text{for all}\ j=0,1,\ldots,2^i-1.
$$
\end{theorem}
The essence of our contribution is that  \emph{coordinate sequences of
a transitive T-function that is uniformly differentiable
modulo 4 are not independent any longer: there are linear relations among
them}.
\section{Main results: statements}
\label{sec:main-st}
Given a transitive T-function $f\colon\mathbb Z_2\to\mathbb Z_2$ and the
initial state $x_0\in\mathbb Z_2$, for $i=0,1,2,\ldots$ denote $x_i=f^i(x_0)$,
$\chi_n^i=\delta_i(f^n(x_0))$, the $n$-th digit in the canonical 2-adic expansion
of the $n$-th iterate of $x_0$. That is, the binary sequence $(\chi_n^i)_{i=0}^\infty$
is the $n$-th coordinate sequence of the recurrence sequence determined by the recurrence
law $x_{i+1}=f(x_i)$.
\subsection{Linear relation}
Our first result yields that if a transitive T-function is uniformly differentiable
modulo 4 then \emph{two} adjacent coordinate sequences satisfy linear relation of form
\eqref{eq.1}:
\begin{theorem}
\label{thm.11}
Let  a transitive T-function $f$ be uniformly differentiable
modulo $4$. Given $x_0\in\mathbb Z_2$, for all $n\geq N_{2}(f)+1$ the following congruence holds:
\begin{equation}
\label{eq:lin_coord}
\chi_n^{i+2^{n-1}}\equiv\chi_{n-1}^i+\chi_{n}^i+
\chi_{n-1}^0+\chi_n^0+\chi_n^{2^{n-1}}
+y(i)\pmod
2. \qquad (i=0,1,2,\ldots),
\end{equation}
The length of the
shortest period of the binary sequence $(y(i))_{i=0}^\infty$ 
is $2^K$,  $0\le K\le {N_{2}(f)}$.
Furthermost, $\gamma(i)$
does not depend on $n$.
\end{theorem}
\begin{proof}
See Appendix \ref{ssec:pr-thm.11}.
\end{proof}
Note that if a T-function is transitive then by Proposition \ref{prop:bij_u-diff}
it is uniformly differentiable modulo 2; so conditions of Theorem \ref{thm.11}
seem not too restrictive: we only demand that the T-function lies in the
second large differentiability class $\mathfrak D_2$ whereas 
it  already
lies in the largest one (i.e., in $\mathfrak D_1$) due to transitivity.

As both polynomial T-functions (the ones represented by polynomials over
$\mathbb Z_2$) and the Klimov-Shamir T-function (of the form $x+(x^2\vee C)$,
$C\in\mathbb Z$) are uniformly differentiable (thus, lie in $\mathfrak D_\infty$
and whence in $\mathfrak D_2$), our Theorem \ref{thm.11} could
be considered as a generalization of results due to Jin-Song Wang and Wen-Feng Qi, \cite{LinPropPol}, and to Molland and Helleseth, \cite{LinearKlSh05,LinPropT-func}. However, the class of transitive T-functions that are uniformly differentiable modulo 4 (thus, the class of T-functions that satisfy our Theorem \ref{thm.11}) is much wider: for instance, it contains all T-functions of forms $f(x)=u(x)+4\cdot
v(x)$
and $f(x)=u(x+4\cdot v(x))$,
where $u$ is a transitive T-function that is uniformly differentiable
modulo 4 and $v$ is an \emph{arbitrary} T-function, see \cite[Proposition
9.29]{AnKhr}. In particular, this implies that a monster T-function from \eqref{eq:wild}
satisfies Theorem \ref{thm.11}.

Moreover,  given an arbitrary T-function  $g$ that is uniformly differentiable modulo 2 (say, given a bijective T-function $g$),  the T-function $f(x)=1+x+2(g(x+1)-g(x))$ is transitive and uniformly differentiable modulo 4; cf. Theorem \ref{thm:expl}.

These examples serve to demonstrate how large is the class of T-functions that
satisfy Theorem \ref{thm.11}. More specific examples of the latter functions
can be constructed with the use of various techniques of non-Archimedean
analysis, see \cite{AnKhr}. For instance, exponential functions of the form $f(x)=ax+a^x$, where
$a\equiv 1\pmod 2$, are uniformly differentiable and transitive, as well
as rational functions of the form $f(x)=\frac{u(x)}{1+4\cdot v(x)}$, where $u$
is a transitive polynomial and $v$ is arbitrary T-function. We remind that
a polynomial over $\mathbb Z_2$ is transitive iff it is transitive modulo 8.
\subsection{Quadratic relation}
Our second result yields that if a T-function  lies in the third largest differentiability
class $\mathfrak D_3$ then there exist a quadratic relation among \emph{three} adjacent coordinate
sequences:
\begin{theorem}
\label{thm:non-lin}
Let the ergodic T-function $f$ be uniformly differentiable modulo $8$. Given
$x_0\in\mathbb Z_2$, for all $n\geq$ $N_{3}(f)+2$ the following congruence
holds:
\begin{equation}
\label{equ.12}
\chi_{n}^{i+2^{n-2}}\equiv\chi_{n-2}^{i}\chi_{n-1}^{i}+\theta(n)(\chi
_{n-2}^{i}+\chi_{n-1}^{i})+\chi_{n}^{i}+y_{i}\pmod2, \qquad (i=0,1,2,\ldots),
\end{equation}
where $\theta(n)\in\{0,1\}$ does not depend on $i$. Furthermost, the length of the shortest period of the binary sequence $(y_{i})_{i=0}^\infty$
 is a factor of $2^{N_{3}(f)}$ if $N_{3}(f)>1$.
\end{theorem}
\begin{proof}
See Appendix \ref{ssec:pr-thm:non-lin}.
\end{proof}
As the Klimov-Shamir T-function $f(x)=x+(x^2\vee C)$ for $C\in\mathbb Z$, is uniformly differentiable,
it satisfy Theorem  \ref{thm:non-lin} once it is transitive, i.e.,
once $C\equiv 5\pmod 8$ or $C\equiv 7\pmod 8$; thus, Theorem \ref{thm:non-lin}
may be considered as a generalization of a result of Yong-Long Luo and Wen-Feng Qui \cite{AlgStr_KlSh} who proved quadratic relation for the Klimov-Shamir T-function.




\section{Application to T-function-based stream ciphers}
\label{sec:Ap-TfSC}
In this section we discuss how  relations \eqref{eq:lin_coord}
and \eqref{equ.12} from Theorems
\ref{thm.11} and \ref{thm:non-lin} may be used to attack stream ciphers that
use T-functions to generate pseudorandom sequences. We do not construct attacks
themselves,
we only point out some approaches that may result in the attacks. We consider mostly the linear relation; however, one
may use the quadratic
relation as well,  by analogy.

Basically a stream cipher is a pseudorandom generator where the produced
binary
sequence is used as a keystream, i.e., is XOR-ed with a plaintext to encrypt
a message. A pseudorandom generator (PRG) can be thought of as an algorithm that
takes at random a short initial binary string, the key, and stretches it
to a much longer binary sequence, the keystream, which looks like random.
that is, passes a set of reasonable tests in a reasonable time.  A stream
cipher must withstand various cryptographic attacks.

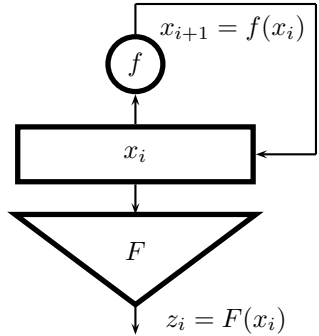
\begin{figure}[h]
\begin{quote}\psset{unit=0.4cm}
 \begin{pspicture}(-5,0)(16,12)
\pscircle[linewidth=2pt](12,10){1}
  \psline(12,11)(12,12)
  \psline{->}(18,7)(16,7)
  \psline(18,12)(18,7)
  \psline(12,12)(18,12)
  \psline{->}(12,8)(12,9)
  \psline{->}(12,2)(12,1)
  \psline{<-}(12,5)(12,6)
  \psframe[linewidth=2pt](8,6)(16,8)
  \pspolygon[linewidth=2pt](8,5)(16,5)(12,2)
  \uput{0}[180](12.4,7){$x_i$}
  \uput{0}[90](12,9.6){$f$}
  \uput{1}[90](12,2.5){$F$}
  \uput{1}[0](11.8,11.2){$x_{i+1}=f(x_i)$}
  \uput{1}[0](12,1.5){$z_i=F(x_i)$}
 \end{pspicture}
\end{quote}
\caption{Pseudorandom generator}
\label{fig:PRG}
\end{figure}
Basically a PRG can be considered as an automaton with no input (see Figure
\ref{fig:PRG}) , where initial
state  $x_0\in\{0,1,\ldots,2^k-1\}$ is a  key, or is produced during the `warming-up' stage from the key
and IV, the initial  vector. We assume that the \emph{state transition function $f$ is a T-function on $k$-bit words}. Moreover, $f$ (as well as the output
function $F$) may depend on a key, or even may change during the encryption
procedure, that is actually the recurrence law is $x_{i+1}=f_i(x_i)$.
In the latter case, the corresponding generator is called \emph{counter-dependent}
\cite{ShTs}; and we assume that \emph{all} $f_i$ are T-functions on $k$-bit
words. Foremost, they may be multivariate T-functions as well, and not
necessarily univariate ones.

Our second basic assumption yields that one knows \emph{sufficiently long  segments
of two coordinate sequences}
$(\chi_{n-1}^i)_{i=0}^\infty$ and $(\chi_{n}^i)_{i=0}^\infty$ $n\le k-1$,
of the state sequence $(x_i)_{i=0}^\infty$. In this Section, we explain \emph{how under these
assumptions one can recover  low order coordinate
sequences $(\chi_m^i)_{i=0}^\infty$ for $m<n-1$}. 
After explaining general method in Subsection \ref{ssec:basic}
for the case of univariate transitive T-function,
we apply the method to  multivariate transitive T-functions (Subsection \ref{ssec:mult}) and to counter-dependent
generators (Subsection \ref{ssec:cnt-dpd}).
\subsection{General method}
\label{ssec:basic}
 Assume that the state transition function $f$
does
not depend on $i$, and assume that $f$ is a reduction modulo $2^k$ of a univariate
transitive
T-function $\tilde f\colon\mathbb Z_2\to\mathbb Z_2$ (i.e., $f=\tilde f\bmod
2^k$) which is \emph{uniformly differentiable
modulo 4} so that $N_2(\tilde f)<n-1< k-1$. It can be shown (see e.g. the
example  at the end of \ref{sssec:TSC}) 
that, given a T-function $f$ which is transitive on $k$-bit
words,
transitive
T-functions $\tilde f\colon\mathbb Z_2\to\mathbb Z_2$ which are uniformly
differentiable modulo 4 and such that $f=\tilde
f\bmod 2^k$ always exist; however, the core of our assumption is that the
number $N_2=N_2(\tilde
f)$
must be sufficiently small: $N_2<n-1< k-1$.

We stress that in most cases the latter assumption is not too restrictive: e.g., for
polynomials with integer coefficients we have that $N_2\le 2$, whereas for the Klimov-Shamir T-function
$x+(x^2\vee 5)$  we have that $N_2=2$; and we have $N_2=1$ for monster T-function
\eqref{eq:wild}. Note that although for  Klimov-Shamir
T-function  $x+(x^2\vee C)$, $C\in\mathbb Z$, which is uniformly differentiable
if $C\in\mathbb N_0$,
the number $N_2$ depends on the length of binary representation of $|C|$,
in practice only small $C$ should be used (e.g., $C=5$) since
distribution properties of the Klimov-Shamir T-function are the poorer the more 1-s are in the
2-adic representation of $C$: For instance, if $C<0$ then  2-dimensional
distribution properties of output sequence of corresponding Klimov-Shamir generator
are practically the same as the ones for the transitive T-function $x\mapsto x-1$,
see \cite{Ryk-KlSh} for a comprehensive  study of distribution properties of  Klimov-Shamir generators; some information about
these can also be found in \cite[Section 11.1]{AnKhr}.

In practice, to construct a T-function $\tilde f$ given the T-function
$f$ we should do absolutely nothing since actually $\tilde f$ is just an
expansion of $f$
to the whole space $\mathbb Z_2$: for instance, if $f$ is a polynomial with
integer coefficients (or Klimov-Shamir T-function $x+(x^2\vee C)$, or monster
T-function \eqref{eq:wild}, etc.), then $\tilde f$ is just the same polynomial
(Klimov-Shamir T-function, monster T-function) considered
over a larger domain, $\mathbb Z_2$ rather than $\mathbb Z/2^k\mathbb Z$.
Thus, our basic assumption just yields that the transitive T-function $f$ must be uniformly differentiable modulo 4 and $N_2(f)$ must be sufficiently
small, at least, smaller than $k-2$; then we can recover coordinate sequences
$(\chi_m^i)_{i=0}^\infty$ for $m=n-2,n-3,\ldots, N_2(f)$. Of course, to recover
the whole $m$-th coordinate sequence we just have to recover its first $2^{m-1}$
terms due to the property \eqref{eq:half-per}.

We now proceed with all these assumptions in mind. 
\subsubsection{The method for a univariate T-function}
\label{sssec:lin-rel}
 We proceed as follows.
\begin{enumerate}
\item Given first $2^n$ bits of coordinate
sequences $(\chi_{n-1}^i)_{i=0}^\infty$ and $(\chi_{n}^i)_{i=0}^\infty$,
we find the sequence $(y(i))_{i=0}^{2^{n-1}-1}$ by solving equations
\eqref{eq:lin_coord} w.r.t. $y(i)$.
\item As by Theorem \ref{thm.11} the sequence $(y(i)$ does not depend on
$n$, having $(y(i))_{i=0}^{2^{n-1}-1}$ and solving equations \eqref{eq:lin_coord}
for $n:=n-1$ and $i=0,1,2,\ldots,2^{n-2}-1$
we find \emph{two}  sequences $\mathcal S_{n-2}^0$ and $\mathcal S_{n-2}^1$ of solutions
$(\chi^i_{n-2})_{i=0}^{2^{n-2}-1}$: the first sequence $\mathcal S_{n-2}^0$
of solutions corresponds
to the choice $\chi_{n-2}^0=0$, whereas the second one, $\mathcal S_{n-2}^1$, corresponds to the
choice $\chi_{n-2}^0=1$ in equation \eqref{eq:lin_coord}. Therefore the two
bit sequences
$\mathcal S_{n-2}^0$ and $\mathcal S_{n-2}^1$
are \emph{mutually complementary},  
$\mathcal S_{n-2}^0\oplus\mathcal S_{n-2}^1=(1)_{i=0}^{2^{n-2}-1}$; that is, the sum of the
$i$-th term of the first sequence with the $i$-th term of the second sequence
is always 1 modulo 2, for all $i=0,1,2,\ldots, 2^{n-2}-1$. Now to find full period
$\mathcal S_{n-2}=(\chi^i_{n-2})_{i=0}^{2^{n-1}-1}$
of the $(n-2)$-th coordinate sequence $(\chi^i_{n-2})_{i=0}^\infty$ we use
relation  \eqref{eq:half-per} (which yields that
$\chi_{n-2}^{i+2^{n-2}}\equiv\chi_{n-2}^i+1\pmod 2$ in the case under consideration)
to continue finite sequences $\mathcal S_{n-2}^0$ and $\mathcal S_{n-2}^1$,
which actually are two variants of the first half-period of the $(n-2)$-th
coordinate sequence $\mathcal S_{n-2}$, to  full periods, of length $2^{n-1}$;
we keep the same notation for these two variants of the full period, i.e.,
$\mathcal S_{n-2}^0$ and $\mathcal S_{n-2}^1$.
Thus we find \emph{two solutions for the full period of $(n-2)$-th coordinate sequence $\mathcal
S_{n-2}$, namely, $\mathcal S_{n-2}^0$ and $\mathcal S_{n-2}^1$, and the
solutions are mutually complementary: $\mathcal S_{n-2}^0\oplus\mathcal S_{n-2}^1=(1)_{i=0}^{2^{n-1}-1}$.}
\item Next, given the sequence $(y(i))$ and two variants   $\mathcal S_{n-2}^0$ and $\mathcal S_{n-2}^1$ of the $(n-2)$-th coordinate   sequence,
we find a pair of mutually complementary sequences $\mathcal S_{n-3}^0$ and $\mathcal S_{n-3}^1$ for either of $\mathcal S_{n-2}^0$ and $\mathcal S_{n-2}^1$
by solving equation \eqref{eq:lin_coord} for  $n:=n-2$ w.r.t. indeterminate
$\chi_{n-3}^i$. However, among these 4 obtained  variants of the first half-period
of the $(n-2)$-th coordinate   sequence there are only \emph{two} different (depending on the value of $\chi^0_{n-2}\oplus\chi^0_{n-3}$) and
they are mutually complementary. Thus, at this step we again obtain
\emph{two solutions, $\mathcal S_{n-3}^0$ and $\mathcal S_{n-3}^1$, for the full period of the $(n-3)$-th coordinate sequence $\mathcal
S_{n-3}$, and the
solutions are mutually complementary: $\mathcal S_{n-3}^0\oplus\mathcal S_{n-3}^1=(1)_{i=0}^{2^{n-2}-1}$.}
\item Proceed with $n:=n-3$, etc.
\end{enumerate}
Two important remarks should be made:
\begin{itemize}
\item As the T-function $f$ is uniformly differentiable modulo 4, at every step $j$ we recover \emph{two} variants of the first
half of a period of the $(n-j)$-th coordinate sequence rather than $2^{2^{n-j}}$
variants for a general transitive T-function $f$, cf. Theorem \ref{thm:half-per};
and the two variants are mutually complementary, so \emph{actually we need to recover
only one of these variants}; so at each step $j$ we just solve $2^{n-j}-1$
linear Boolean equations \eqref{eq:lin_coord}, for $i=1,2,\ldots, 2^{n-j}-1$,
each of one Boolean indeterminate,
$\chi_{n-j}^i$.
\item Nowhere in the algorithm we used the T-function $f$ by itself, e.g.,
its explicit representation in a certain form; \emph{we used only the fact that
$f$ is transitive and uniformly differentiable modulo 4}. 
\end{itemize}
%
\subsection{The case of multivariate T-functions}
\label{ssec:mult}
We firstly stress that a multivariate transitive T-function that is uniformly
differentiable modulo 2 (thus, modulo 4) \emph{does not exist}, see \cite[Theorem
4.51]{AnKhr}; and secondly, that \emph{all known multivariate transitive T-functions
actually are just multivariate representations of univariate transitive T-functions},
see \cite[Section 10.4]{AnKhr}. We briefly explain now what are the latter representations.

A transitive multivariate T-function is a map of form \eqref{eq:t-mult} from the $n$-th Cartesian power $\mathbb
Z_2^n$ of the space $\mathbb Z_2$ to its $m$-th Cartesian power $\mathbb
Z_2^m$ where $m=n$. Loosely speaking, we can
consider an element
of 
$\mathbb Z _2^{m}$ as a table of $m$ one-side infinite
binary rows $x^{(0)},\ldots,x^{(m-1)}$ (say, stretching from left to right). 
To this table, we put into the correspondence  infinite binary string 
(that
is, a 2-adic integer from $\mathbb Z_2$) obtained
by reading successively elements of each column of the table, from top to bottom and from left to right. Thus we establish a one-to-one correspondence
$B$ between
$\mathbb Z_2^m$ and $\mathbb Z_2$. Now, given a transitive univariate
T-function $f$ of form \eqref{eq:t-uni} and using the correspondence, we construct
an $m$-variate transitive T-function $\mathbf f\colon\mathbb Z_2^m\to\mathbb Z_2^m$:
If
$$x=(\chi_0;\chi_1;\chi_2; \ldots)\stackrel{f}{\mapsto}
(\psi_0(\chi_0);\psi_1(\chi_0,\chi_1);\psi_2(\chi_0,\chi_1,\chi_2);\ldots)$$
then $\mathbf f=(h^{(0)},\ldots,h^{(m-1)})$ is defined as follows:
\begin{align}
\label{eq:T-mult-r}
 {x^{(0)}=}&(\chi_0&;{\quad}&\chi_{m}&;{\quad}&\chi_{2m};\quad\ldots&{})&\stackrel{ {h^{(0)}}}{\mapsto}&{}&(\psi_0(x)&;{}\quad&\psi_{m}(x)&;{}\quad&\psi_{2m}(x)\quad;\ldots&)\\
 {x^{(1)}=}&(\chi_1&;{\quad}&\chi_{m+1}&;{\quad}&\chi_{2m+1};\quad\ldots&{})&\stackrel{ {h^{(1)}}}{\mapsto}&{}&(\psi_1(x)&;{}\quad&\psi_{m+1}(x)&;{}\quad&\psi_{2m+1}(x);\quad\ldots&)\nonumber\\
\ldots&\ldots&{}&\ldots&{}&\ldots&{}&\ldots&{}&\ldots&{}&\ldots&{}&\ldots&{}\nonumber\\
 {x^{(m-1)}=}&(\chi_{m-1}&;{\quad}&\chi_{2m-1}&;{\quad}&\chi_{3m-1};\quad\ldots&{})&\stackrel{ {h^{(m-1)}}}{\mapsto}&{}&(\psi_{m-1}(x)&;{}\quad&\psi_{2m-1}(x)&;{}\quad&\psi_{3m-1}(x);\quad\ldots&)\nonumber
\end{align}
where $x^{(0)},\ldots,x^{(m-1)}$ are new 2-adic variables,
$\psi_j(x)=\psi_j(\chi_0,\ldots,\chi_j)$, $j=0,1,2,\ldots$. We stress that known
multivariate
transitive T-functions from  \cite{klimov04new,hong05new} are based on
representations  of this sort of univariate transitive T-functions; and that
these are
multivariate T-functions that are used in the design of
ciphers
Mir-1 \cite{Mir1}, ASC \cite{ASCCipher}, TF-i family \cite{klsh04tfi}, and TSC family \cite{Hong05tsc3}.

To apply our basic approach \ref{sssec:lin-rel} to a multivariate T-function
$\mathbf f$
of this sort, the  corresponding univariate T-function $f$ must be uniformly differentiable
modulo 4. However, even this is not the case, we can consider a \emph{conjugated}
univariate T-function $f^w$ which is uniformly differentiable modulo 4. Indeed,
all univariate transitive T-functions are mutually conjugated: \emph{Given a pair
of transitive T-functions $u,v\colon\mathbb Z_2\to\mathbb Z_2$, there exists
a bijective T-function $w\colon\mathbb Z_2\to\mathbb Z_2$ such that $u=v^w=w^{-1}\circ
v\circ w$}, where $\circ$ stands for  composition of functions (see e.g.
\cite{3-adic}). Now, if we know the conjugating
function $w$ we can apply  method \ref{sssec:lin-rel}.
\subsubsection{The method for multivariate T-functions}
\label{sssec:lin-rel-mult}
Denote $B\colon\mathbb Z_2^m\to\mathbb Z_2$ the above one-to-one correspondence
between $\mathbb Z_2^m$ and $\mathbb Z_2$; thus, given a transitive $m$-variate T-function
$\mathbf f=(h^{(0)},\ldots,h^{(m-1)})\colon\mathbb Z_2^m\to\mathbb Z_2^m$
of form \eqref{eq:T-mult-r}, the corresponding univariate T-function is
$f={\mathbf f}^{B^{-1}}=B\circ \mathbf f\circ B^{-1}$. Now let $g$ be a univariate
T-function for which relations \eqref{eq:lin_coord} holds. As $f=g^w$ for
a suitable T-function $f\colon\mathbb Z_2\to\mathbb Z_2$ (we assume that
$w$ is known), then the $i$-th term $\mathbf x_i$ of the output sequence $(\mathbf x_i)_{i=0}^\infty$ of the generator with the recursion
law $\mathbf x_{i+1}=\mathbf f (\mathbf x_{i})$, $\mathbf x_i=(x^{(0)}_i,\ldots,x^{(m-1)}_i)$
can be represented as
$\mathbf x_i=\mathbf f^i(\mathbf x_0)=B^{-1}(w^{-1}(g^i(w(B(\mathbf x_0)))))$. Therefore, as for $g$  linear relations \eqref{eq:lin_coord} hold,
we can use them to recover coordinate sequences of the sequence
$(\mathbf x_i)_{i=0}^\infty$ since
$$
w(B(\mathbf x_i))=g^i(w(B(\mathbf x_0)).
$$
In other words, \emph{rather than trying to recover coordinate sequences of the
generator with the recursion law $\mathbf x_{i+1}=\mathbf f (\mathbf x_{i})$
and with initial state $\mathbf x_0$ we can study coordinate sequences of
the generator with the recursion law $x_i=g(x_0)$ with the initial state
$x_0=w(B(\mathbf x_0))$ and with a bijective output function $B^{-1}\circ
w^{-1}$}.

Basically the approach will work if the output function $B^{-1}\circ
w^{-1}$ is known. However, the bijective output function $B^{-1}\circ
w^{-1}$ can be considered as ``known"
if $w$ is easy to find and easy to invert; i.e.,
\begin{itemize}
\item if it is easy to find  the conjugating T-function $w$ given T-functions $f$ and $g$ which
are conjugated via $w$: $f=g^w$ (in particular, $w$ must admit then a ``short" representation in some form); and
\item if, given $w$, it is easy to find the inverse
T-function $w^{-1}$ such that $w\circ w^{-1}$ is an identity transformation
(in particular, this means that $w^{-1}$ admits a ``short" representation as well).
\end{itemize}

Indeed,  $B$ is just ``concatenation of columns": it maps $m$ strings (2-adic integers) $x^{(0)},\ldots,x^{(m-1)}$
(see the left side of \eqref{eq:T-mult-r}) to a single string (a 2-adic integer) $x=(\chi_0;\chi_1;\ldots;\chi_{m-1};\chi_{m};\ldots;\chi_{2m-1};\chi_{2m};\ldots)$; so the inverse $B^{-1}$ is just ``cutting a single string into columns of height $m$", which is easy.

Finding $w$ from the equation $f=g^w$ may be an infeasible task: Although, given
two single cycle permutations $f$ and $g$ on some finite set, one may find
all conjugating permutations $w$ by solving the equation by  Cauchy method,
direct application of the latter will take exponentially long time since in our case the set is of order
$2^{km}$ (if we consider an $m$-variate T-function on $k$-bit words).
Also, given a bijective T-function $w$ in some `short' form, there are a number of algorithms
to find the inverse T-function $w^{-1}$; however, the representation
of $w^{-1}$ may be too long and thus the problem of finding $w^{-1}$ will
also be infeasible.

On the other side, in many practical cases main ideas of the approach work
either directly or after certain adjustment: to illustrate, we  apply these to a multivariate T-function from \cite{hong05new}
which is used in TSC family of stream ciphers.
\subsubsection{Linear relation in multivariate function of TSC family of ciphers}
\label{sssec:TSC}
We start with a description   of a general T-function $T$ used in these ciphers.
Given $\mathbf x=(x^{(0)},\ldots,x^{(m-1)})\in\mathbb Z_2^m$, denote $\delta_j(\mathbf x)=(\delta_j(x^{(0)}),\ldots, \delta_j(x^{(m-1)}))$ (the $j$-th columnar binary vector $(\chi_{jm},\ldots,\chi_{(j+1)m-1})$ in the notation of \eqref{eq:T-mult-r})

A special $m$-variate T-function $\alpha(\mathbf x)$ on $k$-bit words (the \emph{odd parameter}) is fixed. For our purposes, we do not need detailed description of  $\alpha(\mathbf x)$, we only note that in our terms $\alpha\colon\mathbb Z_2^m\to\mathbb Z_2$ is a T-function such that $\delta_j(\alpha(\mathbf x))$ does not depend
on $\delta_j(\mathbf x)$ and the Boolean function $\delta_j(\alpha(\mathbf x))$ of Boolean variables $\chi_0,\ldots,\chi_{im-1}$ is of odd weight; that
is $\delta_0(\alpha(\mathbf x))=1$ and algebraic normal form of the Boolean
function $\delta_j(\alpha(\mathbf x))$ contains a monomial $\chi_0\cdots\chi_{jm-1}$
(this is equivalent to the definition of odd parameter in \cite{hong05new,klimov03cryptographic,klimov04new})

Further, an S-box is fixed. That is, the sequence of  permutations $S_0,S_1,S_2,\ldots$
on $m$-bit words is given. Each permutation $S_j$ acts on the $j$-th column
$D_j(\mathbf x)=\delta_j(\mathbf x)=(\chi_{jm},\ldots,\chi_{(j+1)m-1})$ by substituting it for $S_j(D_j(\mathbf x))$ . Also, a sequence
of odd numbers $\sigma_0,\sigma_1,\sigma_2,\ldots$ and a sequence of even numbers $\varepsilon_0,\varepsilon_1,\varepsilon_2,\ldots$
are given. Now the T-function $T$ of TSC family of stream ciphers is defined as follows:
$$
\delta_j(T(\mathbf x))=
\begin{cases}
S_j^{\sigma_j}(\delta_j(\textbf x)),\ \text{if}\ \delta_j(\alpha(\textbf x))=1;\\
S_j^{\varepsilon_j}(\delta_j(\textbf x)),\  \text{if otherwise}.
\end{cases}
$$
The key point is that if $m$ is small, then, given $S_j$ and  a permutation $L_j$
that has the same cycle structure as $S_j$, one easily finds conjugating
permutation $R_j$ by solving the equation $S_j=R_j^{-1}L_jR_j$ by Cauchy
method.

In TSC family $m$ is small: For every TSC$-i$ ($i=1,2,3,4$), the input is arranged into
$m=4$ input words of $k=32$ (TSC-1, -2, -4) or $k=40$ (TSC-3) bits. That
is, to find conjugating permutations one will solve 32 or 40 equations $S_j=R_j^{-1}L_jR_j$
in the symmetric group on 16 elements. Moreover, in TSC family all permutations $S_j$ are
single cycles.

Now put $L_j(z)=(z+1)\bmod{2^m}$, a single cycle permutation that acts on
$m$-bit words by adding 1 modulo $2^m$; that is, $L_j$
reads the $j$-column $(\chi_{im};\chi_{im+1};\ldots;\chi_{(i+1)m-1})$
as a base-2 expansion of a non-negative
integer $z=\chi_jm+\chi_{jm+1}\cdot 2+\cdots+\chi_{(j+1)m-1}2^{m-1}$,
sends $z$ to the least non-negative residue $\overline
{z+1}$ of $z+1$ modulo $2^m$ and returns the column 
$(\delta_0(\overline{z+1});\delta_1(\overline{z+1});\ldots;\delta_{m-1}(\overline{z+1}))$  consider a T-function $L\colon\mathbb Z_2^m\to\mathbb Z_2^m$ defined
as follows:
$$
\delta_j(L(\mathbf x))=
\begin{cases}
L_j^{\sigma_j}(\delta_j(\textbf x)),\ \text{if}\ \delta_j(\alpha(\textbf x))=1;\\
L_j^{\varepsilon_j}(\delta_j(\textbf x)),\  \text{if otherwise}.
\end{cases}
$$
This implies that the T-function $T$ is conjugate to the univariate T-function
$t\colon\mathbb Z_2\to\mathbb Z_2$ that acts as follows: given the input string $x=(\chi_0;\chi_1;\ldots)$, it is considered as concatenation of $m$-bit words $q_0,q_1,\ldots$
$q_j=\chi_{jm}\chi_{im+1}\cdots\chi_{(j+1)m-1}$, the T-function $t$ reads
each word $q_j$
as a base-2 expansion of the non-negative number
$Q_j=\chi_{jm}+\chi_{jm+1}2+\cdots+\chi_{(j+1)m-1}2^{m-1}$, returns
the $m$-bit word $t_j(q_j)$ that is a base-2 expansion of the number
$$
D_j(t(x))=t_j(Q_j)=(Q_j+\sigma_j\cdot a_j(Q_0,\ldots,Q_{j-1})
+\varepsilon_j\cdot(1-a_j(Q_0,\ldots,Q_{j-1})))\bmod
2^m,
$$
where $a_j(Q_0,\ldots,Q_{i-1})=\delta_j(\alpha(B^{-1}(x)))$, $B$ is the one-to-one
correspondence
between $\mathbb Z_2^m$ and $\mathbb Z_2$ from \ref{sssec:lin-rel-mult}.

It turns out that coordinate sequences of each sequence $(D_j(t^i(x)))_{i=0}^\infty$
of $m$-bit words satisfy relation \eqref{eq:lin_coord}. Note that our claim
is that the
relation holds \emph{only} within every sequence $(D_j(t^i(x)))_{i=0}^\infty$,
and \emph{not} necessarily between the coordinate sequences
$(\delta_{jm-1}(t^i(x)))_{i=0}^\infty$ and $(\delta_{jm}(t^i(x)))_{i=0}^\infty$
since they belong to coordinate sequences of different sequences, of $(D_{j-1}(t^i(x)))_{i=0}^\infty$ and
$(D_j(t^i(x)))_{i=0}^\infty$, respectively.

To prove the claim it suffices to prove it for coordinate sequences (of sufficiently
large order) of a univariate T-function $f$ that is defined as follows. Let
$u\colon\mathbb Z/2^k\mathbb Z\to\mathbb Z/2^k\mathbb Z$ is a transitive
T-function on $k$-bit words, let the map $v\colon\mathbb Z/2^k\mathbb Z\to\{0,1\}$
takes value 1 on the odd number of $k$-bit words: $\#\{z\in\mathbb Z/2^k\mathbb Z\colon v(z)=1\}$ is odd; let $\sigma$ be odd, and let $\varepsilon$ be even.
Given $x\in\mathbb Z_2$, $x$ admits a unique representation $x=\bar x+2^k\tilde
x$ for a suitable $\tilde x\in\mathbb Z_2$. Now put
$$
f(x)=u(\bar x)+2^k(\tilde x+(\sigma-\varepsilon)v(\bar x)+\varepsilon).
$$

Firstly, we note that $f$ is uniformly differentiable and that $N_2(f)\le
k$. Indeed, given $h=2^\ell r$ for $\ell\ge k$, one has
$f(x+h)=u(\bar x)+2^k(\tilde x + 2^{\ell-k}r+(\sigma-\varepsilon)v(\bar x)+\varepsilon)=
f(x)+2^\ell r=f(x)+h$.

Secondly, $f$ is transitive. Indeed,
$$
f^{2^{k}}(x)=u^{2^k}(\bar x)+2^k\left(\tilde x+
(\sigma-\varepsilon)\sum_{j=0}^{2^k-1}v(u^j(\bar x))+2^k\varepsilon\right);
$$
however, $s=\sum_{j=0}^{2^k-1}v(u^j(\bar x))$ is odd by the definition of
$v$ since $u^j(\bar x)$
runs through all $k$-bit words as $j=0,1,2,\ldots, 2^k-1$, due to transitivity
of $u$. Thus, $f$ is transitive modulo $2^{k+2}$ as the map
$\tilde x\mapsto \tilde x+(\sigma-\varepsilon)s+2^k\varepsilon$ is obviously
transitive modulo 4 as $(\sigma-\varepsilon)s+2^k\varepsilon$ is odd. Finally,
$f$ is transitive by Theorem \ref{thm:erg} and thus satisfy conditions of
Theorem \ref{thm.11}. This proves our claim (of course, the transitivity
of $f$ might be proved directly rather than by applying Theorem \ref{thm:erg}).

We stress that we only state that there are linear relations of form \eqref{eq:lin_coord}
in the output sequences of generators based on  T-functions of the sort
of ones
used in TSC stream ciphers, and we do \emph{not} claim that these relations
 affect (or do not affect) the security of the ciphers. The latter is out of scope of the paper; it worth noting here only that the ciphers were successfully
attacked, however, using  vulnerabilities other than the ones we  indicate, see e.g. \cite{LinearTSC05,Zhang08tsc4dif}.

It also worth noticing here that the method can not be immediately applied
to stream ciphers Mir-1, TF-i and ASC although all of these are based on
a multivariate version of Klimov-Shamir T-function $x+(x^2\vee C)$ for which
the relations hold due to the result of Molland and Helleseth mentioned at
the beginning of the paper.

\subsection{The case of counter-dependent generators}
\label{ssec:cnt-dpd}
A counter-dependent generator is a pseudorandom generator with the recursion
law $x_{i+1}=f_i(x_i)$, that is, the state transition (and/or the output)
function changes dynamically during processing. Counter-dependent generators
were introduced in \cite{ShTs}; in \cite[Section 10.3]{AnKhr} it is shown
that counter-dependent generators can be considered as wreath products of
dynamical systems which are ordinary generators, and the corresponding theory is developed. The theory enables
one to construct counter-dependent generators of the longest possible period.
Generators of this kind were used in ABC stream ciphers, see
 \cite{ABCv1,ABCv2,ABCv206props,ABCv206safe,ABCv3}.

Loosely speaking, wreath product of generators is a cascaded composition of generators,
see Figure \ref{fig:wrPRG}.
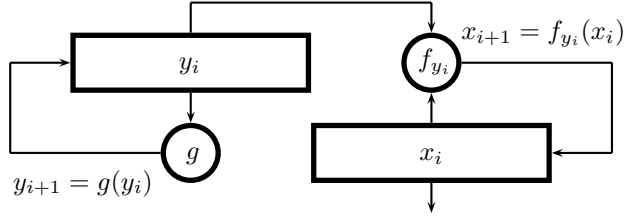
\begin{figure}[h]
\begin{quote}\psset{unit=0.4cm}
 \begin{pspicture}(-5,5)(16,13)
\pscircle[linewidth=2pt](12,10){1}
  \psline(4,11)(4,12)
  \psline{<-}(12,11)(12,12)
  \psline{->}(18,7)(16,7)
  \psline(18,10)(18,7)
  \psline(13,10)(18,10)
  \psline{->}(12,8)(12,9)
  \psline{<-}(12,5)(12,6)
  \psframe[linewidth=2pt](8,6)(16,8)
  \psframe[linewidth=2pt](0,9)(8,11)
  \pscircle[linewidth=2pt](4,7){1}
  \psline{->}(-2,10)(0,10)
  \psline(-2,7)(3,7)
  \uput{0}[180](4.3,6.9){$g$}
  \uput{1}[0](-2.9,6){$y_{i+1}=g(y_i)$}
  \psline(-2,7)(-2,10)
  \psline(4,12)(12,12)
  \psline{<-}(4,8)(4,9)
  \uput{1}[0](2.6,9.9){$y_i$}
  \uput{0}[180](12.4,6.9){$x_i$}
  \uput{0}[90](12.1,9.6){$f_{y_i}$}
  \uput{1}[0](12,11){$x_{i+1}=f_{y_i}(x_i)$}
 \end{pspicture}
\end{quote}
\caption{Counter-dependent generator, the wreath product of generators}
\label{fig:wrPRG}
\end{figure}
If all $f_{y_i}$ are T-functions on $k$-bit words, the maximum length of the shortest period of the counter-dependent generator from Figure \ref{fig:wrPRG} is $p\cdot 2^k$, where $p$ is the length of the shortest period of the generator
with the recursion law $y_{i+1}=g(y_i)$. For conditions when the counter-dependent
 generator achieves the longest
possible period see \cite[Theorem 10.9; Lemma 10.12]{AnKhr}; structure of the corresponding output sequence
is presented at Figure \ref{fig:wr-seq}: the shortest period of this sequence achieves
the maximum length, $p\cdot
2^k$, i.e., the period is a finite sequence $(x_i)_{i=0}^{p2^k-1}$ of length $p\cdot 2^k$ of $k$-bit
words which is a union of $p$ subsequences $(x_{r+pj})_{j=0}^{2^k-1}$, $r=0,1,2,\ldots,p-1$,
and each subsequence $(x_{r_pj})_{j=0}^{2^k-1}$ is generated by a transitive
T-function $w_r$:
$w_r=f_{y_{r+p+1}}\circ\cdots \circ f_{y_{r}}$,
  $w_r(x_{r+(\ell-1)p})=x_{r+\ell p}$,
  $\ell=1,2,\ldots$. We conclude now that \emph{if all T-functions $f_{y_j}$ are uniformly differentiable
modulo 4} then all T-functions $w_r$ are uniformly differentiable modulo
4 and transitive;
thus, all T-functions $w_r$ satisfy conditions of Theorem \ref{thm.11}. Therefore
\emph{coordinate sequences of every subsequence $(x_{r+\ell p})_{\ell=0}^\infty$
of output sequence $(x_i)_{i=0}^\infty$
satisfy linear relation \eqref{eq:lin_coord}}.

\begin{figure}[ht]
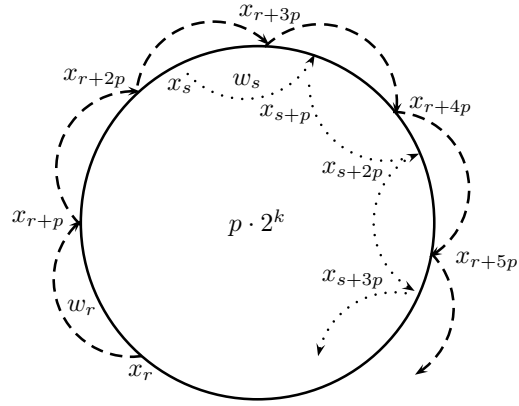

\psset{unit=0.35}
\pspicture(3.684904,-18.115097)(59.165097,-1.034904)
\scalebox{0.35 -0.35}
{

\rput(35,58){
\scalebox{2.8 -2.8}{
\rput(16.2,8.8){$x_{s+3p}$}
\rput(16.2,13){$x_{s+2p}$}
\rput(13.6,15.2){$x_{s+p}$}
\rput(9.5,16.3){$x_{s}$}
}}

\rput(35,58){
\scalebox{2.8 -2.8}{
\rput(4,11.1){$x_{r+p}$}
\rput(08,05.2){$x_{r}$}
\rput(21.4,9.5){$x_{r+5p}$}
\rput(19.6,15.5){$x_{r+4p}$}
\rput(13,19.1){$x_{r+3p}$}
\rput(06.2,16.6){$x_{r+2p}$}
}}

\rput(35,58){
\scalebox{2.8 -2.8}{
\rput(12.1,16.5){$w_s$}
\rput(05.7,7.7){$w_r$}
\rput(12.5,11.0){$p\cdot2^{k}$}
\pscircle[linewidth=1pt](12.5,11.0){6.9}
\psarc[linewidth=1pt,linestyle=dotted]{->}(12,18.8){3}{216}{335}
\psarc[linewidth=1pt,linestyle=dotted]{->}(17.5,16.4){3}{175}{295}
\psarc[linewidth=1pt,linestyle=dotted]{->}(20,11){3}{120}{242}
\psarc[linewidth=1pt,linestyle=dotted]{->}(17.8,5.3){3}{75}{170}
\psarc[linewidth=1pt,linestyle=dashed]{<-}(17.1,7.7){3}{300}{45}
\psarc[linewidth=1pt,linestyle=dashed]{<-}(17.7,12.3){3}{299}{89}
\psarc[linewidth=1pt,linestyle=dashed]{<-}(14.9,15.7){3}{350}{135}
\psarc[linewidth=1pt,linestyle=dashed]{<-}(10.8,15.8){3}{44}{170}
\psarc[linewidth=1pt,linestyle=dashed]{<-}(7.7,13.1){3}{85}{227}
\psarc[linewidth=1pt,linestyle=dashed]{<-}(7.6,8.8){3}{130}{278}
}}

}\endpspicture
  \caption{Structure of the sequence generated by  wreath product.}
  \label{fig:wr-seq}
\end{figure}
It is worth noting here that the above result on linear relations in coordinate
sequences
produced by wreath
products of generators can not be applied immediately to ABC stream ciphers
since the latter use wreath products of linear feedback shift register with
an `add-xor' generator. However, the latter is based on a transitive T-function of the form
$(\ldots((x\oplus a_1)+a_2)\oplus a_3)+a_4)\oplus\cdots$ which is not uniformly
differentiable modulo 4. Of course, this does not serve a  proof (or a disproof)
that there are no linear relations between coordinate sequences produced by the ABC wreath product.
\section{Conclusion}
\label{sec:Concl}
In the paper, we prove that a vast body of transitive T-functions exhibit linear and quadratic
weaknesses: we found a linear (Theorem \ref{thm.11}) and a quadratic (Theorem \ref{thm:non-lin}) relation that are satisfied by output sequences generated
by  univariate transitive T-functions that constitute a very vast class $\mathfrak D_2$
(see Subsection \ref{ssec:T-Calc} about the latter class).  Earlier relations of this sort
were known only for  T-functions of two special types: for the Klimov-Shamir
T-function $x+(x^2\vee C)$ and for polynomials with integer coefficients. The class $\mathfrak
D_2$ is much wider: it contains rational functions, exponential functions
as well as their various compositions with bitwise logical operations.
Moreover, we proved that  relations of this kind
hold in output sequences of corresponding classes of multivariate
T-functions as well as in output sequences of T-function-based counter-dependent
generators; the latter are generators with a recursion law of the form $x_{i+1}=f_i(x_i)$.
Primitives of both types, the multivariate T-function-based ordinary
generators and T-function-based counter-dependent generators, are used in
stream ciphers, e.g., in  ASC, TF-i, TSC, and in ABC.
We illustrated our method by finding linear relations for T-function of the
sort used in TSC stream ciphers.

\appendix
\section{Proofs of Theorems \ref{thm.11} and \ref{thm:non-lin}}
During the proofs, we will need the following
\begin{lemma}
\label{le:prod}
Let $f$ be a transitive T-function, and let $f$ be uniformly differentiable modulo
4, then
$$
\left(f^{2^{N_2(f)}}(z)\right)^\prime_2\equiv
\prod_{j=0}^{2^{N_2(f)}-1}f^\prime_2(f^j(z))\equiv 1\pmod 4,
$$
for every $z\in\mathbb Z_2$.
\end{lemma}
\begin{proof}[Proof of Lemma \ref{le:prod}]
The left-side congruence immediately follows from the chain rule; the right-side
congruence is proved in \cite{AnKhr}, see the end of the proof of Theorem
4.55 there. It is worth noticing that we actually prove both congruences
while proving Theorem \ref{thm:non-lin},  see Step 5 in the proof of the
latter.
\end{proof}
\subsection{Proof of Theorem \ref{thm.11}}
\label{ssec:pr-thm.11}
From the transitivity of the T-function $f$ (see Definition \ref{def:trans})
it follows that%
$f^{2^{n-1}}(x)\equiv x\pmod{2^{n-1}}$; that is
\begin{equation}
\label{eq:f-iter}
f^{2^{n-1}}(x)=x+2^{n-1}\varphi(x)
\end{equation}
for a suitable map $\varphi\colon\mathbb{Z}_{2}\rightarrow\mathbb{Z}_{2}$. As $f$ is uniformly differentiable modulo 4, from \eqref{eq:f-iter} we
deduce that%
\begin{equation}
\label{eq:main}
f^{i+2^{n-1}}(x)=f^{i}(f^{2^{n-1}}(x))=
f^{i}(x+2^{n-1}\varphi(x))
\equiv f^{i}(x)+2^{n-1}\varphi(x)(f^{i}(x))_{2}^{^{\prime}}%
\pmod{2^{n+1}}%
\end{equation}
once $n\geq N_{2}(f)+1$.

Further, $\varphi(x)\equiv\alpha(x)+2\beta(x)\pmod4$ where
$\alpha\colon\mathbb Z_2\to\mathbb F_2=\{0,1\}$. We claim that $\alpha(x)=1$
for all $x\in\mathbb Z_2$.
Indeed, if otherwise,  then \eqref{eq:f-iter} implies that
\[
f^{2^{n-1}}(x)=x+2^{n}\beta(x)\equiv x\pmod{2^{n}},%
\]
in a contradiction to the transitivity of $f$ as necessarily $f^{2^{n-1}}(x)\not\equiv x\pmod{2^{n}}$ whenever $f$ is transitive, see Definition \ref{def:trans}. Thus, given $x\in\mathbb Z_2$,
\begin{equation}
\label{eq:phi}
\varphi(x)\equiv1+2\beta
\pmod4,
\end{equation}
for a suitable $\beta=\beta(x)\in\mathbb Z_2$.

As $f$ is bijective,
$f_{2}^{^{\prime}}(x)\equiv1\pmod2$ for all $x\in\mathbb{Z}_{2}%
$, see Proposition \ref{prop:bij_u-diff}. This in view of \eqref{eq:main}
and \eqref{eq:phi} implies that if we denote $(f^{i}(x))_{2}^{^{\prime}}\equiv1+2\gamma
\pmod4$ for a suitable $\gamma=\gamma(i;x)\in\{0,1\}$, then
\begin{equation}
\label{equ.8}
f^{i+2^{n-1}}(x)  
\equiv f^{i}(x)+2^{n-1}(1+2\beta)(1+2\gamma)\pmod{2^{n+1}}
\equiv f^{i}(x)+2^{n-1}+2^{n}(\beta+\gamma)\pmod{2^{n+1}}.
\end{equation}
Remind that $\chi_{j}^{\ell}=\delta_{j}(f^{\ell}(x))\in\{0,1\}$ ($j,\ell=0,1,2,\ldots$) according to our notation. 
With the notation, given $x=x_0\in\mathbb Z_2$, the transitivity of $f$ implies that
\begin{multline}
\label{equ.9}
f^{2^{n-1}}(\chi_{0}^{0}+\chi_{1}^{0}\cdot2+\cdots)\equiv\chi_{0}^{2^{n-1}%
}+\chi_{1}^{2^{n-1}}\cdot2+\cdots+\chi_{n-1}^{2^{n-1}}\cdot2^{n-1}+\chi
_{n}^{2^{n-1}}\cdot2^{n}\equiv\\
\chi_{0}^{2^{n-1}}+\chi_{1}^{2^{n-1}}\cdot2+\cdots+\chi_{n-2}^{2^{n-1}%
}\cdot2^{n-2}+(\chi_{n-1}^{0}\oplus1)\cdot2^{n-1}+\chi_{n}^{2^{n-1}}%
\cdot2^{n}\pmod{2^{n+1}}, %
\end{multline}
where $\oplus$ stands for addition modulo 2. On the other hand, %
\[
f^{2^{n-1}}(\chi_{0}^{0}+\chi_{1}^{0}\cdot2+\cdots)\equiv\chi_{0}^{0}+\chi
_{1}^{0}\cdot2+\cdots+\chi_{n}^{0}\cdot2^{n}+2^{n-1}+2^{n}\beta
\pmod{2^{n+1}}%
\]
in view of \eqref{eq:f-iter} and \eqref{eq:phi}. Comparing both congruences,
we conclude that
$
\chi_{n}^{2^{n-1}}\equiv\chi_{n-1}^{0}+\chi_{n}^{0}+\beta\pmod 2
$;
finally,
\begin{equation}
\beta\equiv\chi_{n-1}^{0}+\chi_{n}^{0}+\chi_{n}^{2^{n-1}}\pmod2.
\label{equ.10}%
\end{equation}
Now from(\ref{equ.8}), (\ref{equ.9}), (\ref{equ.10}) we obtain:
\begin{multline*}
\chi_{0}^{i+2^{n-1}}+\chi_{1}^{i+2^{n-1}}\cdot2+\cdots+\chi_{n-1}%
^{i+2^{n-1}}\cdot2^{n-1}+\chi_{n}^{i+2^{n-1}}\cdot2^{n}
\equiv\\
\chi_{0}^{i}+\chi_{1}^{i}\cdot2+\cdots+\chi_{n}^{i}\cdot2^{n}%
+2^{n-1}+(\chi_{n-1}^{0}+\chi_{n}^{0}+\chi_{n}^{2^{n-1}}+\gamma)2^{n}%
\pmod{2^{n+1}};%
\end{multline*}
henceforth, %
\begin{equation}
\chi_{n}^{i+2^{n-1}}\equiv\chi_{n-1}^{i}+\chi_{n}^{i}+\chi_{n-1}^{0}+\chi
_{n}^{0}+\chi_{n}^{2^{n-1}}+\gamma\pmod2. \label{equ.11}%
\end{equation}
Note that the term  $\chi^i_{n-1}$ occurs in the right side due to the carry.

Now take (and fix) arbitrary $x=x_0\in\mathbb Z_2$. We claim that the function $y(i)=\gamma(i;x)$ is   periodic with respect to the variable $i=0,1,2,\ldots$, and that the length of
the shortest period of $y(i)$ is a factor of $2^{N_2(f)}$.

Denote $N=N_2(f)$. As $y(\ell)=\delta_{1}%
((f^{\ell}(x))_{2}^{^{\prime}})$ by the definition, $y(\ell)$ can not
depend on $n$ once $n\ge N+1$; furthermost,  we have that
$y(i+2^N)=\delta_{1}((f^{i+2^N}(x))_2^\prime)$.
Using sequentially chain rule and Lemma \ref{le:prod} for $z=f^{i}(x)$ we
get:
\begin{equation*}
(f^{i+2^N}(x))_{2}^{^{\prime}}\equiv
\prod\limits_{j=0}^{i+2^N-1}f_{2}^{^{\prime}}(f^{j}(x))\equiv
\prod\limits_{j=0}^{i-1}f_{2}^{^{\prime}}(f^{j}(x))
\prod\limits_{j=0}^{2^N-1}f_{2}^{^{\prime}}(f^{j+i}(x))\equiv
\prod\limits_{j=0}^{i-1}f_{2}^{^{\prime}}(f^{j}(x))\equiv (f^{i}(x))_{2}^{^{\prime}} \pmod4.
\end{equation*}
Therefore, $y(i+2^N)=\delta_{1}((f^{i+2^N}(x))_2^\prime)=\delta_{1}((f^{i}(x))_2^\prime)=y(i)$.
This proves our claim and Theorem \ref{thm.11}.\qed
\subsection{Proof of Theorem \ref{thm:non-lin}}
\label{ssec:pr-thm:non-lin}

The proof mimics respective steps of the proof of Theorem \ref{thm.11}.

\underline{Step 1:} As $f^{2^{n-2}}(x)=x+2^{n-2}\varphi(x)$
for  a suitable map $\varphi\colon\mathbb{Z}_{2}\rightarrow\mathbb{Z}_{2}$,
given $n\geq N_{3}(f)+2$
we have that
\begin{equation}
\label{eq:main1}
f^{i+2^{n-2}}(x)
\equiv f^{i}(x)+2^{n-2}\varphi(x)(f^{i}(x))_{3}^{^{\prime}}\pmod{2^{n+1}},
\end{equation}
cf. \eqref{eq:f-iter} and \eqref{eq:main}.

\underline{Step 2:} Denote $\varphi(x)\equiv\alpha+2\beta+4\gamma(\operatorname{mod}8)$, for
suitable $\alpha,\beta,\gamma\in\{0,1\}$. We prove that $\alpha=1$ 
exactly in the same way as in the proof of Theorem \ref{thm.11}.

\underline{Step 3:} We have then that $(f^{i}(x))_{3}^{^{\prime}}=1+2\lambda+4\eta
\pmod8$, for suitable $\lambda,\eta\in\{0,1\}$. Therefore,
\begin{multline}
\label{equ.14}
f^{i+2^{n-2}}(x)=f^{i}(f^{2^{n-2}}(x))=f^{i}(x+2^{n-2}\varphi
(x))
\equiv \\
f^{i}(x)+2^{n-2}+2^{n-1}(\beta+\lambda)+2^{n}(\beta\lambda
+\gamma+\eta)\pmod{2^{n+1}}, %
\end{multline}
cf. \eqref{equ.8}.

\underline{Step 4:} Now we act as in the proof of \eqref{equ.11}. On the
one hand,
\begin{multline}
f^{2^{n-2}}(\chi_{0}^{0}+\chi_{1}^{0}\cdot2+\cdots)\equiv\chi_{0}^{2^{n-2}%
}+\chi_{1}^{2^{n-2}}\cdot2+\cdots+\chi_{n-1}^{2^{n-2}}\cdot2^{n-1}+\chi
_{n}^{2^{n-2}}\cdot2^{n}\\
\equiv\chi_{0}^{2^{n-2}}+\chi_{1}^{2^{n-2}}\cdot2+\cdots+(\chi_{n-2}^{0}%
\oplus1)\cdot2^{n-2}+\chi_{n-1}^{2^{n-2}}\cdot2^{n-1}+\chi_{n}^{2^{n-2}}%
\cdot2^{n} \pmod{2^{n+1}}, \label{equ.15}%
\end{multline}
while on the other hand,
\[
f^{2^{n-2}}(\chi_{0}^{0}+\chi_{1}^{0}\cdot2+\cdots)\equiv\chi_{0}^{0}+\chi
_{1}^{0}\cdot2+\cdots+\chi_{n}^{0}\cdot2^{n}+2^{n-2}+2^{n-1}\beta+2^{n}%
\gamma\pmod{2^{n+1}}.%
\]
From here we deduce that
$
\chi_{n-1}^{2^{n-2}}=\chi_{n-1}^{0}\oplus\chi_{n-2}^{0}\oplus\beta;
$
henceforth
\begin{equation}
\beta\equiv\chi_{n-1}^{0}+\chi_{n-2}^{0}+\chi_{n-1}^{2^{n-2}}%
\pmod2 \label{equ.16},%
\end{equation}
cf. \eqref{equ.10}. Now, combining together
(\ref{equ.14}),(\ref{equ.15}), (\ref{equ.16}), we get
\begin{multline*}
\chi_{0}^{i+2^{n-2}}+\chi_{1}^{i+2^{n-2}}\cdot2+\cdots+\chi_{n-1}%
^{i+2^{n-2}}\cdot2^{n-1}+\chi_{n}^{i+2^{n-2}}\cdot2^{n}
\equiv\\
\chi_{0}^{i}+\chi_{1}^{i}\cdot2+\cdots+\chi_{n}^{i}\cdot2^{n}%
+2^{n-2}+2^{n-1}(\chi_{n-2}^{0}+\chi_{n-1}^{0}+\chi_{n-1}^{2^{n-2}}%
+\lambda)+2^{n}(\beta\lambda+\gamma+\eta)\pmod{2^{n+1}};%
\end{multline*}
so we conclude that
\begin{equation*}
\chi_{n-1}^{i+2^{n-2}}\equiv\chi_{n-2}^{i}+\chi_{n-1}^{i}+\chi_{n-2}^{0}%
+\chi_{n-1}^{0}+\chi_{n-1}^{2^{n-2}}+\lambda\pmod2  
\end{equation*}
and that
\[
\chi_{n}^{i+2^{n-2}}\equiv\chi_{n-2}^{i}\chi_{n-1}^{i}+\chi_{n-2}^{i}%
(\chi_{n-2}^{0}+\chi_{n-1}^{0}+\chi_{n-1}^{2^{n-2}}+\lambda)+\chi_{n-1}%
^{i}(\chi_{n-2}^{0}+\chi_{n-1}^{0}+\chi_{n-1}^{2^{n-2}}+\lambda)+\chi_{n}%
^{i}+\beta\lambda+\gamma+\eta\pmod2. 
\]
From here we finally obtain that
\[
\chi_{n}^{i+2^{n-2}}\equiv\chi_{n-2}^{i}\chi_{n-1}^{i}+\theta(n)(\chi
_{n-2}^{i}+\chi_{n-1}^{i})+\chi_{n}^{i}+y_{i}\pmod2,
\]
where
$\theta(n)\equiv\chi_{n-2}^{0}+\chi_{n-1}^{0}+\chi_{n-1}^{2^{n-2}}+\lambda
\pmod
2$ and
$y_i\equiv\beta\lambda+\gamma+\eta\pmod 2$.

\underline{Step 5:} Take and fix arbitrary $x\in\mathbb Z_2$ and $n\ge N_3(f)+2$;
therefore we fix $\beta, \gamma\in\{0,1\}$, however, both $\beta$ and $\gamma$
depend on $n$. We claim that the binary sequence $(y_{i})_{i=0}^\infty$ is periodic, and
that the length of its shortest period is a factor of $2^{N_{3}(f)}$.

Indeed, by the chain rule
\begin{equation}
\label{eq:chain3}
\left(f^{\ell}(z)\right)_{3}^{^{\prime}}\equiv\prod\limits_{j=0}^{\ell-1}f_{3}^{^{\prime}%
}(f^{j}(x))\pmod8,
\end{equation}
for arbitrary $z\in\mathbb Z_2$ and $\ell=1,2,\ldots$.
As $f$ is a transitive T-function, 
$f^{i+2^{N_{3}(f)}}(x)= f^{i}(x+2^{N_{3}(f)}\Phi(x))$ for a suitable
$\Phi\colon\mathbb Z_2\to\mathbb Z_2$ (cf. \eqref{eq:f-iter} and \eqref{eq:main});
and moreover,
$$
f^{j}(x+2^{N_{3}(f)}\Phi(x))=f^{j}
(x)\bmod2^{N_{3}(f)}+2^{N_{3}(f)}\widetilde\Phi_j(x),
$$
where $f^{j}(x)\bmod2^{N_{3}(f)}$ stands for the least non-negative residue
of $f^{j}(x)$
modulo $2^{N_{3}(f)}$  and $\widetilde\Phi_j(x)\in\mathbb Z_2$.
Now combining the latter equality with 
\eqref{eq:chain3}
we see that
\begin{multline}
\label{eq:gamma-per}
\left(f^{i+2^{N_{3}(f)}}(x)\right)^\prime_3=
\left(f^{i}(x+2^{N_{3}(f)}\Phi(x))\right)^\prime_3\equiv\prod\limits_{j=0}^{i-1}f_{2}^{^{\prime}}\left(f^{j}(x+2^{N_{3}%
(f)}\Phi(x))\right)
\equiv\\
\prod\limits_{j=0}^{i-1}f_{2}^{^{\prime}}\left(f^{j}%
(x)\bmod2^{N_{3}(f)}+2^{N_{3}(f)}\widetilde\Phi_j(x)\right)
\equiv
\prod\limits_{j=0}^{i-1}f_{2}^{^{\prime}}\left(f^{j}%
(x)\bmod2^{N_{3}(f)}\right)
\equiv \left(f^i(x)\right)_3^\prime\pmod 8,
\end{multline}
as $f_{2}^{^{\prime}}(x)$ is a periodic function with a period of length
$2^{N_{3}(f)}$, cf. Proposition \ref{prop:diff_mod}.

Now, as
$\lambda=\delta_{1}((f^{i}%
(x))_{3}^{^{\prime}})$ and $\eta=\delta_{2}((f^{i}(x))_{3}^{^{\prime}}%
)$ 
the functions $\lambda=\lambda(i)$ and $\eta=\eta(i)$ are periodic with respect to the
variable $i=0,1,2\ldots$, 
and lengths of their shortest periods are
factors of $2^{N_{3}(f)}$. Consequently, the sequence $(y_i)_{i=0}^\infty$ is periodic, and the length of its shortest period is $2^K$ for some $0\le K\le
N_3(f)$.\qed

\bigskip


\begin{thebibliography}{99}

\bibitem{me-NATO}
V.~Anashin.
\newblock Non-{A}rchimedean theory of {T}-functions.
\newblock In {\em Proc. Advanced Study Institute Boolean Functions in
  Cryptology and Information Security}, volume~18 of {\em NATO Sci. Peace
  Secur. Ser. D Inf. Commun. Secur.}, pages 33--57, Amsterdam, 2008. IOS Press.

\bibitem{me-CJ}
V.~Anashin.
\newblock Non-{A}rchimedean ergodic theory and pseudorandom generators.
\newblock {\em The Computer Journal}, 53(4):370--392, 2010.

\bibitem{AnKhr}
V.~Anashin and A.~Khrennikov.
\newblock {\em Applied Algebraic Dynamics}, volume~49 of {\em de Gruyter
  Expositions in Mathematics}.
\newblock Walter~de~Gruyter GmbH \& Co., Berlin---N.Y., 2009.

\bibitem{me:1}
V.~S. Anashin.
\newblock Uniformly distributed sequences of $p$-adic integers.
\newblock {\em Mathematical Notes}, 55(2):109--133, 1994.

\bibitem{me:ex}
V.~S. Anashin.
\newblock Uniformly distributed sequences in computer algebra, or how to
  constuct program generators of random numbers.
\newblock {\em J. Math. Sci.}, 89(4):1355--1390, 1998.

\bibitem{me:2}
V.~S. Anashin.
\newblock Uniformly distributed sequences of $p$-adic integers, {II}.
\newblock {\em Discrete Math. Appl.}, 12(6):527--590, 2002.

\bibitem{me:conf}
Vladimir Anashin.
\newblock Uniformly distributed sequences over $p$-adic integers.
\newblock In I.~Shparlinsky A.~J. van~der Poorten and H.~G. Zimmer, editors,
  {\em Number theoretic and algebraic methods in computer science. Proceedings
  of the Int'l Conference (Moscow, June--July, 1993)}, pages 1--18. World
  Scientific, 1995.

\bibitem{ABCv3}
Vladimir Anashin, Andrey Bogdanov, and Ilya Kizhvatov.
\newblock {ABC}: {A} {N}ew {F}ast {F}lexible {S}tream {C}ipher, {V}ersion 3.
\newblock Technical report, eSTREAM, 2005.
\newblock Available from
  \url{http://www.ecrypt.eu.org/stream/p2ciphers/abc/abc_p2.pdf}.

\bibitem{ABCv2}
Vladimir Anashin, Andrey Bogdanov, and Ilya Kizhvatov.
\newblock {ABC}: {A} {N}ew {F}ast {F}lexible {S}tream {C}ipher, {V}ersion 2.
\newblock Technical report, eSTREAM, 2005.
\newblock Available from \url{http://crypto.rsuh.ru/papers/abc-spec-v2.pdf}.

\bibitem{ABCv1}
Vladimir Anashin, Andrey Bogdanov, Ilya Kizhvatov, and Sandeep Kumar.
\newblock {ABC} : {A} {N}ew {F}ast {F}lexible {S}tream cipher.
\newblock Technical Report 2005/001, eSTREAM, 2005.
\newblock Available from \url{http://eprint.iacr.org/}.

\bibitem{ABCv206safe}
Vladimir Anashin, Andrey Bogdanov, Ilya Kizhvatov, and Sandeep Kumar.
\newblock {ABC} {I}s {S}afe {A}nd {S}ound.
\newblock {\em Cryptology ePrint Archive}, 2006.
\newblock Available from
  \url{http://www.ecrypt.eu.org/stream/papersdir/079.pdf}.

\bibitem{ABCv206props}
Vladimir Anashin, Andrey Bogdanov, Ilya Kizhvatov, and Sandeep Kumar.
\newblock {S}ecurity and {I}mplementation {P}roperties of {ABC} v.2.
\newblock Technical Report 2006/026, eSTREAM, 2006.
\newblock Available from
  \url{http://www.ecrypt.eu.org/stream/papersdir/2006/026.pdf}.

\bibitem{3-adic}
F.~Durand and F.~Paccaut.
\newblock Minimal polynomial dynamics on the set of 3-adic integers.
\newblock {\em Bull. London Math. Soc.}, 41(2):302--314, 2009.

\bibitem{hong05new}
J.~Hong, D.~Lee, Y.~Yeom, and D.~Han.
\newblock A new class of single cycle {T}-functions.
\newblock In {\em Fast Software Encryption}, volume 3557 of {\em Lect. Notes
  Comp. Sci.}, pages 68--82. Springer-Verlag, 2005.

\bibitem{Hong05tsc3}
Jin Hong, Dong~Hoon Lee, Yongjin Yeom, and Daewan Han.
\newblock {T}-function based stream cipher {TSC}-3.
\newblock Technical Report 2005/031, eSTREAM, 2005.
\newblock Available from
  \url{http://www.ecrypt.eu.org/stream/ciphers/tsc3/tsc3.pdf}.

\bibitem{klimov03cryptographic}
A.~Klimov and A.~Shamir.
\newblock Cryptographic applications of {T}-functions.
\newblock In {\em {S}elected {A}reas in {C}ryptography}, volume 3006, pages
  248--261, 2003.

\bibitem{KlSh}
A.~Klimov and A.~Shamir.
\newblock A new class of invertible mappings.
\newblock In B.S.Kaliski~Jr.et al., editor, {\em Cryptographic Hardware and
  Embedded Systems 2002}, volume 2523 of {\em Lect. Notes in Comp. Sci}, pages
  470--483. Springer-Verlag, 2003.

\bibitem{klimov04new}
A.~Klimov and A.~Shamir.
\newblock New cryptographic primitives based on multiword {T}-functions.
\newblock In {\em Fast Software Encryption}, 2004.

\bibitem{klsh04tfi}
A.~Klimov and A.~Shamir.
\newblock {T}he {TF}-i family of stream ciphers.
\newblock Handout distributed at: The State of the Art of Stream Ciphers –
  SASC, 2004.

\bibitem{klimov05thesis}
Alexander Klimov.
\newblock {\em Applications of T-functions in Cryptography}.
\newblock PhD thesis, Weizmann Institute of Science, 2005.
\newblock Available from \url{http://www.wisdom.weizmann.ac.il/~ask/}.

\bibitem{klimov05app}
Alexander Klimov and Adi Shamir.
\newblock New applications of {T}-functions in block ciphers and hash
  functions.
\newblock In Henri Gilbert and Helena Handschuh, editors, {\em Fast Software
  Encryption}, volume 3557 of {\em Lecture Notes in Computer Science}, pages
  18--31. Springer, 2005.

\bibitem{Kobl}
N.~Koblitz.
\newblock {\em $p$-adic numbers, $p$-adic analysis, and zeta-functions},
  volume~58 of {\em Graduate texts in math.}
\newblock Springer-Verlag, second edition, 1984.

\bibitem{Koloko}
N.~Kolokotronis.
\newblock Cryptographic properties of nonlinear pseudorandom number generators.
\newblock {\em Designs, Codes and Cryptography}, 46:353--363, 2008.

\bibitem{kotomina99}
L.~Kotomina.
\newblock Fast nonlinear congruential generators.
\newblock Master's thesis, {R}ussian {S}tate {U}niversity for the {H}umanities,
  {M}oscow, 1999.
\newblock In {R}ussian.

\bibitem{Lar}
M.~V. Larin.
\newblock Transitive polynomial transformations of residue class rings.
\newblock {\em Discrete Mathematics and Applications}, 12(2):141--154, 2002.

\bibitem{vestwebsite}
Synaptic~Laboratories Limited.
\newblock The {VEST} cryptosystem for semiconductors.
\newblock \url{http://www.vestciphers.com/en/index.html}.

\bibitem{AlgStr_KlSh}
Yong~Long Luo and Wen-Feng Qui.
\newblock On the algebraic structure of {K}limov-{S}hamir {T}-function.
\newblock {\em Journal on Communications}, 29(10), 2008.
\newblock In Chinese.

\bibitem{Mah}
K.~Mahler.
\newblock {\em $p$-adic numbers and their functions}.
\newblock Cambridge Univ. Press, 1981.
\newblock (2nd edition).

\bibitem{Mir1}
Alexander Maximov.
\newblock A new stream cipher {M}ir-1.
\newblock Technical Report 2005/017, eSTREAM, 2005.
\newblock Available from http://www.ecrypt.eu.org/stream.

\bibitem{LinearKlSh05}
H{\aa}vard Molland and Tor Helleseth.
\newblock A linear weakness in the {K}limov-{S}hamir {T}-function.
\newblock In {\em Proc. 2005 IEEE Int. Symp. on Information Theory}, pages
  1106--1110, 2005.

\bibitem{LinPropT-func}
H{\aa}vard Molland and Tor Helleseth.
\newblock Linear properties in {T}-functions.
\newblock {\em IEEE Trans. Inf. Theory}, 52(11):5151--5157, 2006.

\bibitem{tsc4Cipher}
Dukjae Moon, Daesung Kwon, Daewan Han, Jooyoung Lee, Gwon~Ho Ryu, Dong~Wook
  Lee, Yongjin Yeom, and Seongtaek Chee.
\newblock {T}-function based stream cipher {TSC}-4.
\newblock Technical Report 2006/024, eSTREAM, 2006.
\newblock Available from
  \url{http://www.ecrypt.eu.org/stream/papersdir/2006/024.pdf}.

\bibitem{LinearTSC05}
F.~Muller and T.~Peyrin.
\newblock Linear cryptanalysis of the {TSC} family of stream ciphers.
\newblock In {\em ASIACRYPT}, volume 3788 of {\em Lect. Notes. Comp. Sci.},
  pages 373--394. Springer, 2005.

\bibitem{vest06phase2Cipher}
Sean O'Neil, Benjamin Gittins, and Howard Landman.
\newblock {VEST}.
\newblock Technical report, eSTREAM, 2006.
\newblock Available from \url{http://www.ecrypt.eu.org/stream/vestp2.html}.

\bibitem{Ryk-KlSh}
S.~V. Rykov.
\newblock On properties of {K}limov-{S}hamir pseudorandom number generator.
\newblock {\em Discrete Math. Appl.}, 2011.
\newblock In press.

\bibitem{Sch}
W.~H. Schikhof.
\newblock {\em Ultrametric calculus}.
\newblock Cambridge University Press, 1984.

\bibitem{ShTs}
A.~Shamir and B.~Tsaban.
\newblock Guaranteeing the diversity of number generators.
\newblock {\em Information and Computation}, 171:350--363, 2001.

\bibitem{LinPropPol}
Jin-Song Wang and Wen-Feng Qi.
\newblock Linear equation on polynomial single cycle {T}-function.
\newblock In Dingyi~Pei et~al., editor, {\em Inscrypt 2007}, volume 4990 of
  {\em Lect. Notes Comp. Sci.}, pages 256--270, Berlin--Hedelberg, 2008.
  Springer Verlag.

\bibitem{ASCCipher}
Kai-Thorsten Wirt.
\newblock {ASC} – {A} {S}tream {C}ipher with {B}uilt–{I}n {MAC}
  {F}unctionality.
\newblock In {\em Proc. World Acad. Sci. Engineering and Technology},
  volume~23, 2007.

\bibitem{Zhang08tsc4dif}
Haina Zhang and Xiaoyun Wang.
\newblock Differential cryptanalysis of {T}-function based stream cipher
  {TSC}-4.
\newblock In Kil-Hyun Nam and Gwangsoo Rhee, editors, {\em ICISC}, volume 4817
  of {\em Lect. Notes Comp. Sci.}, pages 227--238. Springer, 2007.

\bibitem{LinCompT-fun}
Wenying Zhang and Chuan-Kun Wu.
\newblock The algebraic normal form, linear complexity and k-error linear
  complexity of single-cycle {T}-function.
\newblock In G.~Gong et~al., editor, {\em SETA 2006}, volume 4086 of {\em Lect.
  Notes Comp. Sci.}, pages 391--401, Berlin--Heidelberg, 2006. Springer-Verlag.

\end{thebibliography}
\end{document}